\documentclass[journal,twocolumn]{IEEEtran} 
\IEEEoverridecommandlockouts
\usepackage{amssymb}
\usepackage{amsmath}
\usepackage{graphicx}
\usepackage{graphicx,subfig}
\usepackage{amsfonts}
\usepackage{lscape}
\usepackage{color}
\usepackage{ dsfont }
\usepackage{algpseudocode}
\newtheorem{theorem}{Theorem}
\newtheorem{construction}{Construction}

\newtheorem{lemma}{Lemma}
\newtheorem{corollary}{Corollary}
\newtheorem{definition}{Definition}
\newtheorem{remark}{Remark}
\newtheorem{example}{Example}
\newtheorem{observation}{Observation}

\newcommand{\beqno}{ \begin{equation*} }
\newcommand{\eeqno}{ \end{equation*} }
\newcommand{\beq}{ \begin{equation} }
\newcommand{\eeq}{ \end{equation} }

\newcommand{\calM}{\mathcal{M}}
\newcommand{\calC}{\mathcal{C}}
\newcommand{\calS}{\mathcal{S}}
\newcommand{\bbf}{\mathbb{F}}

\begin{document}

\title{Fractional repetition codes with flexible repair from combinatorial designs }
\author{Oktay Olmez and Aditya Ramamoorthy

\thanks{This work was supported in part by the NSF under grants CCF-1320416, CCF-1149860, CCF-1116322 and DMS-1120597 and by TUBITAK project numbers 114F246 and 115F064. The material in this work has appeared in part at the 50th Annual Allerton Conference on Communication, Control and Computing, 2012, the 2013 International Symposium on Network Coding and the 2013 Asilomar Conference on Signals, Systems and Computers.

Oktay Olmez (oolmez@ankara.edu.tr) is with the Department of Mathematics at Ankara University, Tandogan, Ankara, Turkey. Aditya Ramamoorthy (adityar@iastate.edu) is with the Department of Electrical and Computer Engineering at Iowa State University, Ames, IA 50011. Copyright (c) 2014 IEEE. Personal use of this material is permitted.  However, permission to use this material for any other purposes must be obtained from the IEEE by sending a request to pubs-permissions@ieee.org.}
}

\maketitle
\begin{abstract}
Fractional repetition (FR) codes are a class of regenerating codes for distributed storage systems with an exact (table-based) repair process that is also uncoded, i.e., upon failure, a node is regenerated by simply downloading packets from the surviving nodes. In our work, we present constructions of FR codes based on Steiner systems and resolvable combinatorial designs such as affine geometries, Hadamard designs and mutually orthogonal Latin squares. The failure resilience of our codes can be varied in a simple manner. We construct codes with normalized repair bandwidth ($\beta$) strictly larger than one; these cannot be obtained trivially from codes with $\beta = 1$. Furthermore, we present the Kronecker product technique for generating new codes from existing ones and elaborate on their properties. FR codes with locality are those where the repair degree is smaller than the number of nodes contacted for reconstructing the stored file. For these codes we establish a tradeoff between the local repair property and failure resilience and construct codes that meet this tradeoff. Much of prior work only provided lower bounds on the FR code rate. In our work, for most of our constructions we determine the code rate for certain parameter ranges.
\end{abstract}
\section{Introduction}
\label{sec:intro}
Large scale data storage systems that are employed in social networks, video streaming websites and cloud storage are becoming increasingly popular. In these systems, the integrity of the stored data and the speed of the data access needs to be maintained even in the presence of unreliable storage nodes. This issue is typically handled by introducing redundancy in the storage system, through the usage of replication and/or erasure coding. However, the large scale, distributed nature of the systems under consideration introduces another issue. Namely, if a given storage node fails, it need to be regenerated so that the new system continues to have the properties of the original system. 
It is of course desirable to perform this regeneration in a distributed manner and optimize performance metrics associated with the regeneration process. Firstly, one would like to ensure that the regeneration process be fast. For this purpose we would like to minimize the data that needs to be downloaded from the surviving nodes. Moreover, we would like the surviving nodes and the new node to perform very little (ideally no) computation, as this also induces a substantial delay in the regeneration process that is comparable to the download time (since nowadays, memory access bandwidth is comparable to network bandwidth \cite{inside_ssd_book}). In addition, the regeneration induces a workload on the surviving storage nodes and it is desirable to perform the regeneration by connecting to a small number of nodes. Connecting to a small set of nodes also reduces the overall energy consumption of the system.

\begin{figure*}[t]
\centering
\null\hfill
\subfloat[]{
  \centering
  \label{DSS-(5,3,4)}
  \includegraphics[scale=0.7]{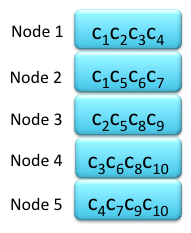}
  }
  \hfill
\subfloat[]{
  \label{multisymbol}
  \centering
  \raisebox{2mm}{
  \includegraphics[scale=0.7]{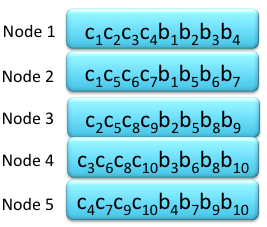}}}
  \hfill\null
\caption{(a) A DSS with $(n,k,d,\alpha) = (5,3,4,4)$. Each node contains a subset of size $4$ of the packets from $\{c_1, \dots, c_{10}\}$. First node for instance contains symbols $c_i, i = 1, \dots, 4$ that is $\{c_1, c_2, c_3, c_4\}$. When there is no confusion we simply use the notation $c_1c_2c_3c_4$ instead of the set notation $\{c_1, c_2, c_3, c_4\}$. (b) The DSS is constructed by applying a $(20, 18)$-MDS code followed by the inner fractional repetition code shown in the figure. It is specified with parameters $(5,3,4,8)$. When a node fails we contact the remaining four nodes and download two packets from each to repair the failed node. This DSS can be obtained from the $(5,3,4,4)$ DSS on the left by trivial $\beta$-expansion, where $\beta = 2$.}
\end{figure*}

In recent years, codes which are designed to satisfy the needs of data storage systems have been the subject of much investigation and there is extensive literature on this topic. 
Depending upon the specific metrics that are optimized there are different requirements that the distributed storage system needs to satisfy. However, broadly speaking, all systems have the following general characteristics. A  distributed storage system (henceforth abbreviated to DSS) consists of $n$ storage nodes, each of which stores $\alpha$ packets (we use symbols and packets interchangeably). A given user, also referred to as the data collector needs to have the ability to reconstruct the stored file by contacting any $k$ nodes; this is referred to as the maximum distance separability (MDS) property of the system. To ensure reliability in the system, the DSS also needs to repair a failed node. This is accomplished by contacting a set of $d$ surviving nodes and downloading $\beta$ packets from each of them for a total repair bandwidth of $\gamma = d \beta$ packets. Thus, the system has a repair degree of $d$, normalized repair bandwidth $\beta$ and total repair bandwidth $\gamma$. The new DSS should continue to have the MDS property.

A simple technique for obtaining a DSS is to treat the file that needs to be stored as a set of symbols over a large enough finite field, generate encoded symbols by using an MDS code (such as a Reed-Solomon (RS) code) and then store each encoded symbol on a different storage node. It is well recognized that the drawback of this method is that upon failure of a given storage node, a large amount of data needs to be downloaded from the remaining storage nodes (equivalent to recreating the file). To address this issue, the technique of regenerating codes was developed in the work of Dimakis et al. \cite{dimakis2010}. In the framework of \cite{dimakis2010}, the repair degree $d \geq k$ and the system needs to have the property that a failed node can be repaired from {\it any} set of $d$ surviving nodes.
The principal idea of regenerating codes is to use subpacketization. In particular, one treats a given physical block as consisting of multiple symbols (unlike the MDS code that stores exactly one symbol in each node). Coding is now performed across the packets such that the file can be recovered by contacting a certain minimum number of nodes. In addition, one can regenerate a failed node by downloading appropriately coded data from the surviving nodes. The work of \cite{dimakis2010} identified a fundamental tradeoff between the amount of storage at each node and the amount of data downloaded for repairing a failed node under the mechanism of functional repair, where the new node is functionally equivalent to the failed node, though it may not be an exact copy of it. Two points on the curve deserve special mention and are arguably of the most interest from a practical perspective. The minimum bandwidth regenerating (MBR) point refers to the point where the repair bandwidth, $\gamma$ is minimum. Likewise, the minimum storage regenerating (MSR) point refers to the point where the storage per node, $\alpha$ is minimum. 

In a different line of work, it has been argued that repair bandwidth is not the only metric for evaluating the repair process. It has been observed that the number of nodes that are contacted for purposes of repair is also an important metric that needs to be considered. The model of \cite{dimakis2010}, which enforces repair from any set of $d$ surviving nodes requires $d$ to be at least $k$. The notion of local repair was introduced in \cite{gopalan2012, papD12, oggier2011}, and considers the design of DSS where $d < k$. However, one only requires that there is some set of $d$ surviving nodes from which the repair can take place.

The majority of work in the design of codes for DSS considers {\it coded} repair where the surviving nodes and the new node need to compute linear combinations of the stored symbols for regeneration. It is well recognized that the read/write bandwidth of machines is comparable to the network bandwidth \cite{inside_ssd_book}.
Thus, this process induces additional undesirable delays \cite{jiekak2013} in the repair process. The process can also be potentially memory intensive since the packets comprising the file are often very large (of the order of GB).  Motivated by these issues, reference \cite{el2010} considered the following variant of the DSS problem. The DSS needs to satisfy the property of {\it exact} and {\it uncoded} repair, i.e., the regenerating node needs to produce an exact copy of the failed node by simply downloading packets from the surviving nodes. This allows the entire system to work without requiring any computation at the surviving nodes. In addition, they considered systems that are resilient to multiple $(> 1)$ failures. However, the DSS only has the property that the repair can be conducted by contacting some set of $d$ nodes, i.e., unlike the original setup, repair is not guaranteed by contacting any set of $d$ nodes. This is reasonable as most practical systems operate via a table-based repair, where the new node is provided information on the set of surviving nodes that it needs to contact. The work of \cite{el2010} proposed a construction whereby an outer MDS code is concatenated with an inner ``fractional repetition" code that specifies the placement of the coded symbols on the storage nodes. The main challenge here is to design the inner fractional repetition (FR) code in a systematic manner.

In this work, we present several families of FR codes and analyze their properties.
This paper is organized as follows. In Section \ref{sec:back_rel_work}, we outline our precise problem formulation, elaborate on the related work in the literature and summarize the contributions of our work. We discuss our FR code constructions for the case when $d \geq k$ in Section \ref{sec:resolv_design_dss}, and explain the Kronecker product technique in Section \ref{sec:iterative_construction}. The locally recoverable FR codes where $d<k$ are considered in Section \ref{sec:locally_recoverable_dss} and Section \ref{sec:conclusions_future_work} outlines the conclusions and opportunities for future work.

\begin{figure*} [t]
\centering
\includegraphics[scale=0.32]{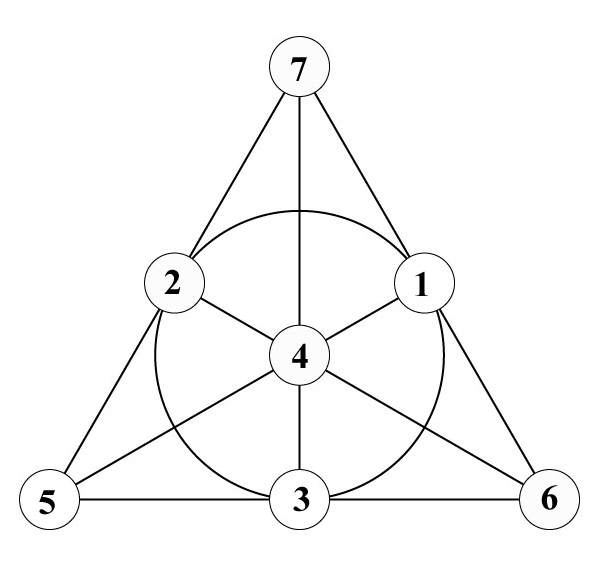}
\includegraphics[scale=0.52]{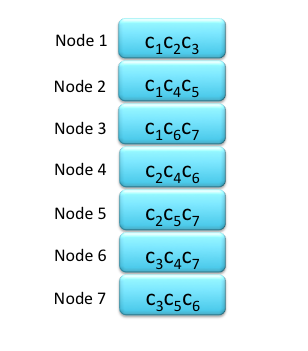}
\caption{\label{fanoplane1} $(7, 3, 3, 1)-$BIBD also known as the Fano plane. Nodes of the DSS, which can be obtained from the Fano plane, are listed on the right. }
\end{figure*}

\section{Background, Related Work and Summary of Contributions}
\label{sec:back_rel_work}
A DSS is specified by parameters $(n,k,d,\alpha)$ where $n$ - number of storage nodes, $k$ - the minimum number of nodes to be contacted for recovering the file, $d$ - the number of nodes to be contacted in order to regenerate a failed node and $\alpha$ -  the storage capacity. In case of repair, the new node downloads $\beta$ packets from each surviving node, for a total of $\gamma = d \beta$ packets. Let $\calM$ denote the size of file being stored on the DSS. 
We consider the design of fractional repetition codes that are best explained by means of the following example \cite{rashmi2009} with $(n,k,d, \alpha) = (5,3,4,4)$.

\begin{example}
\label{example:complete_graph}
Consider a file of $\calM = 9$ packets $(a_1, \dots, a_9) \in \mathbb{F}_q^9$ that needs to be stored on the DSS. We use a $(10,9)$ MDS code that outputs $10$ packets $c_i = a_i, i = 1, \dots, 9$ and $c_{10} = \sum_{i=1}^9 a_i$. The coded packets $c_1, \dots, c_{10}$ are placed on $n=5$ storage nodes as shown in Fig. \ref{DSS-(5,3,4)}. This placement specifies the inner fractional repetition code. It can be observed that each $c_i$ is repeated $\rho = 2$ times and the total number of symbols $\theta = 10$. Any user who contacts any $k=3$ nodes can recover the file (using the MDS property). Moreover, a failed node can be regenerated by downloading one packet each from the four surviving nodes, i.e., $\beta =1 $ and $d = 4$, so that $\gamma = 4$.
\end{example}

Thus, the approach uses an MDS code to encode a file consisting of a certain number of symbols. Let $\theta$ denote the number of encoded symbols. Copies of these symbols are placed on the $n$ nodes such that each symbol is repeated $\rho$ times and each node contains $\alpha$ symbols. Moreover, if a given node fails, it can be exactly recovered by downloading $\beta$ packets from some set of $d$ surviving nodes, for a total repair bandwidth of $\gamma=d\beta$. It is to be noted that in this case $\alpha = \gamma$, i.e., these schemes operate at the MBR point. In the example above, $\beta = 1$, so that $\alpha = d$. One can also consider systems with $\beta > 1$ in general. A simple way to do this is replicating the symbols in the storage system. The resultant DSS has the parameters $(n,k,d,\beta \alpha)$ with $\beta>1$. However, in this work we show that there are infinite families of FR codes with $\beta >1$ which cannot be obtained this way. In Fig. \ref{multisymbol} we illustrate the DSS obtained by replicating the $(5,3,4,4)$-DSS when $\beta=2$.

\begin{figure*} [t]
\centering
\includegraphics[scale=0.65]{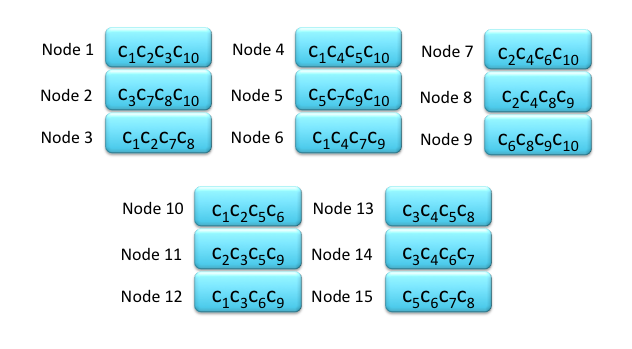}
\caption{The figure shows a DSS where $n=15,k=4, d=2,\theta=10,\alpha =4,\rho =6$. A node can be repaired by contacting the other two nodes in the same column. The system is resilient up to $5$ node failures.}
\label{fig:non_trivial_large_beta}
\end{figure*}

Before introducing the formal definition of a fractional repetition (FR) code we need the notion of $\beta$-recoverability. Let $[n]$ denote the set $\{1, 2, \dots, n\}$.
\begin{definition}[\textbf{$\beta$-recoverability}]
Let $\Omega = [\theta]$ and $V_i, i = 1, \dots, d$ be subsets of $\Omega$. Let $V = \{V_1, \dots, V_d\}$ and consider $A \subset \Omega$ with $|A| = d\beta$. We say that $A$ is \textit{$\beta$-recoverable} from $V$ if there exist $B_i \subseteq V_i$ for each $i= 1, \dots, d$ such that $B_i \subset A, |B_i| = \beta$ and $\displaystyle \cup_{i=1}^d B_i = A$.
\end{definition}

\begin{definition}[\textbf{FR Codes}]
\label{defn:fr_code}
A \textit{fractional repetition} (FR) code $\calC = (\Omega, V)$ for a $(n,k,d,\alpha)$-DSS with repetition degree $\rho$ and normalized repair bandwidth $\beta = \alpha/d$  ($\alpha$ and $\beta$ are positive integers) is a set of $n$ subsets $V=\{V_1, \dots, V_n\}$ of a symbol set $\Omega = [\theta]$ with the following properties.
\begin{itemize}
\item[(a)] The cardinality of each $V_i$ is $\alpha$.
\item[(b)] Each element of $\Omega$ is contained in exactly $\rho$ sets in $V$.
\item[(c)] Let $V^{surv}$ denote any $(n- \tau)$ sized subset of $V$ and $V^{fail} = V \setminus V^{surv}$. Each $V_j \in V^{fail}$ is $\beta$-recoverable from some $d$-sized subset of $V^{surv}$. Let $\rho_{res}$ be the maximum value of $\tau$ such that this property holds.
\end{itemize}

We provide the following example to illustrate that requirement (c) of Definition \ref{defn:fr_code} plays an important role in our study.

\begin{example} \label{eg:beta_rec_cond_fr} Consider the sets $\Omega=\{1,2,3,4,5,6\}$, and two different families of subsets of $\Omega$ as shown below.
\begin{align*}
V &=\{\{1,2,3\}, \{2,3,4\},\{4,5,6\},\{1,5,6\}\}, \text{~and}\\
W &=\{\{1,2,3\}, \{3,4,5\},\{2,5,6\},\{1,4,6\}\}.
\end{align*}
Both $V$ and $W$ satisfy the requirements (a) and (b) of Definition \ref{defn:fr_code}. However, note that $\{1,2,3\} \cap \{4,5,6\}=\emptyset$. This implies that $\{1,2,3\}$ is not $1$-recoverable from the set $$\{\{2,3,4\},\{4,5,6\},\{1,5,6\}\}.$$ So $C=(\Omega, V)$ cannot be a fractional repetition code. In contrast, any failed set in $W$ is $1$-recoverable and thus $C=(\Omega, W)$ is a fractional repetition code with $\delta=1$.
\end{example}

The value of $\rho_{res}$ is a measure of the resilience of the system to node failures, under the constraint of exact and uncoded repair.
The \textit{file size} is given by
\begin{align*}
\displaystyle \mathcal{M} =  \min_{I \subset [n],|I| = k} |\cup_{i \in I} V_i|
\end{align*}
and the \textit{code rate} is defined as $\displaystyle R_{\calC}=\frac{\mathcal{M}}{n\alpha}$. We emphasize that $R_{\calC}$ depends on $k$.
\end{definition}

Note that the parameters of a FR code are such that $\theta \rho = n \alpha$. Thus, the code rate $R_{\calC} = \frac{\calM}{n\alpha} \leq \frac{\theta}{n\alpha} = \frac{1}{\rho}$. Moreover as $\rho \geq 2$, the maximum rate of any FR code is at most $\frac{1}{2}$. It is to be noted that the parameter $\calM$ also sets the code rate of the outer MDS code; it is exactly $\calM/\theta$. For a FR code $\calC = (\Omega, V)$ and an index set $\mathcal{I} \subseteq [n]$, we say that nodes $V_i \in V$ for $i \in \mathcal{I}$ cover $\zeta$ symbols if $\zeta = |\cup_{i \in \mathcal{I}} V_i|$.

\begin{figure*}[t]
\centering
\includegraphics[scale=0.65]{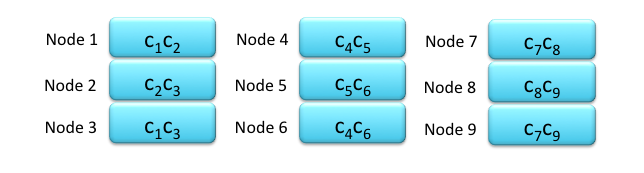}
\caption{ \label{fig:kron_local_eg} A failed node can be recovered by contacting two nodes and downloading one packet from each. The code is resilient up to five failures and the file size is $5$. The minimum distance is of the code is $6$, since any four nodes can recover the file. }
\end{figure*}

The work of \cite{el2010}, only considered FR codes with $\beta = 1$ and $k \leq d$, i.e., for recovery the new node would contact $d$ surviving nodes and download a single packet from each of them. For their codes, the requirement (c) in Definition \ref{defn:fr_code} is satisfied and the system is resilient to $\rho - 1$ failures, i.e., $\rho_{res} = \rho - 1$.  It is to be noted that the requirement of $d \geq k$ is essential in the problem formulation considered in \cite{dimakis2010} since the systems require node recovery from any set of $d$ surviving nodes. In that setup if $d < k$, it is easy to see that one can always specify a failed node and a set of $d$ nodes from which recovery is impossible. However, in the framework of \cite{el2010}, the recovery requirement is relaxed. Specifically, to recover from a failure, the new node contacts a specific set of nodes from which it regenerates the failed node. Thus, the recovery process is table-based and for each node we only need to guarantee the existence of one set of $d$ nodes from which recovery is possible. 
Thus, it becomes possible to have systems with $d<k$. In fact, in Section \ref{sec:locally_recoverable_dss}, of this paper, we present several constructions of FR codes where $d < k$. In the literature, these are referred to as codes that allow for local repair.

For FR codes, the failure resilience $\rho_{res}$ and the code rate $R_\calC$ are two evaluation metrics and it is evident that there is a tradeoff between them. Indeed, if the outer MDS code does not add any redundancy, i.e., $\calM = \theta$ then $k$ would need to be chosen such that any $k$ nodes cover all the $\theta$ symbols and the code rate of the system would be exactly $\frac{\theta}{n \alpha}$. However, in this case the DSS will be resilient to at most $\rho - 1$ failures under any possible recovery procedure, i.e., even without any constraint on the repair.
In contrast, if the outer code introduces nontrivial redundancy, the file size $\calM$ would be lower but it may be possible to reconstruct the DSS in the presence of more than $\rho - 1$ failures. To see this, consider Example \ref{example:complete_graph} where the outer MDS code has rate $9/10$. Note that under exact and uncoded repair, this DSS is resilient to only one failure. However, the DSS can be reconstructed even in the presence of the failure of any two nodes, since any three surviving nodes cover at least nine symbols. Our proposed codes will also be evaluated in terms of their {\it minimum distance} which quantifies this tradeoff. 
\begin{definition}[\textbf {Minimum Distance of a DSS}]
The \textit{minimum distance} of a DSS denoted $d_{\min}$ is defined to be the size of the smallest subset of storage nodes whose failure guarantees that the file is not recoverable from the surviving nodes.
\end{definition}
The Singleton bound on the minimum distance in this context can be found, e.g., in eq. (15) in reference \cite{kamath_et_al_14}.
\begin{lemma}[\textbf{Singleton Bound}]\label{singleton} Consider a DSS with parameters $(n,k,d,\alpha)$ with file size $\calM$ and minimum distance $d_{\min}$. Then,
\begin{align*}
d_{\min} \leq n- \left \lceil \frac{\calM}{\alpha}  \right \rceil +1.
\end{align*}
\end{lemma}

It turns out that codes that have the local repair property, i.e., codes with $d < k$ suffer a penalty on the maximum possible minimum distance. This tradeoff was captured in the case of scalar (i.e., $\alpha = 1$) codes by
 \cite{gopalan2012} and by \cite{papD12} in the case of vector (i.e., $\alpha > 1$) codes.
\begin{lemma}\label{local_bound} Consider a DSS with parameters $(n,k,d,\alpha)$ with file size $\calM$ and minimum distance $d_{\min}$. Then,
\begin{align*}
d_{\min} \leq n- \left \lceil \frac{\calM}{\alpha}  \right \rceil -\left \lceil \frac{\calM}{d\alpha} \right \rceil + 2.
\end{align*}
\end{lemma}

We note that if $d \geq k$, we have $\lceil \frac{\calM}{d\alpha}\rceil = 1$ so that the bound above reduces to the Singleton bound.
\begin{observation}\label{obs:meet_bound}
A given DSS meets the Singleton bound if $k = \lceil \frac{\calM}{\alpha} \rceil$. Similarly, a code meets the bound in Lemma \ref{local_bound} if $k = \left \lceil \frac{\calM}{\alpha}  \right \rceil +\left \lceil \frac{\calM}{d\alpha} \right \rceil - 1$.
\end{observation}

It is to be noted that the bound in Lemma \ref{local_bound} holds for all possible local repair codes. In this work, we consider the added constraint that the repair takes place purely by download. Thus, for our constructions, the bound in Lemma \ref{local_bound} is in general loose. In Section \ref{sec:locally_recoverable_dss} we derive a tighter upper bound on the minimum distance of codes where the repair process is local and operates purely by download.

At various points we will need to use the well-known inclusion-exclusion principle for computing the maximum file sizes that can be supported by our DSS. For the sake of completeness, we state the result here.
\begin{theorem}\label{thm:inc_enc} [\textbf{Inclusion-Exclusion principle}]
Consider $n$ sets $A_1, A_2, \dots, A_n$. If $\mathcal{I} \subseteq [n]$, let $A_{\mathcal{I}} = \cap_{j \in \mathcal{I}} A_j$. Then
\begin{align}
|A_1 \cup A_2 \cup \dots \cup A_n| = \sum_{\emptyset \neq \mathcal{I} \subseteq [n]} (-1)^{|\mathcal{I}| + 1} |A_{\mathcal{I}}|.
\end{align}
It can also be shown that
\begin{align} \label{eq:inc_enc_lower_bd}
|A_1 \cup A_2 \cup \dots \cup A_n| \geq \sum_{i=1}^n |A_i| - \sum_{i < j} |A_i \cap A_j|.
\end{align}
\end{theorem}


%
%

\begin{table*}[t]
\begin{center}
\small
\centering
\begin{tabular}{| p{1.5cm} |  p{2.5cm} |  p{1.5cm} | p{1.5cm} | p{8cm} |}
\hline
Method & $(n,\theta, \alpha, \rho)$ & $\mathcal{M}$ & Range of $k$ & Comments\\
\hline
\hline
Steiner Systems with $t=2$ & $(n,\theta,\alpha, \frac{\theta-1}{\alpha -1})$ & ${\tiny \geq k\alpha - \binom{k}{2}}$ & ${\tiny 1 \leq k \leq \alpha}$ & Steiner systems with $\alpha=3,4,5$ are completely characterized and explicit constructions are known. Here we list the necessary and sufficient conditions for the cases $\alpha=3,4,5$
\begin{itemize}   \item Steiner systems with $\alpha=3$ exists for any $\theta \equiv 1,3 \mod 6$.
\item Steiner systems with $\alpha=4$ exists for any $\theta \equiv 1,4 \mod 12$.
\item Steiner systems with $\alpha=5$ exists for any $\theta \equiv 1,5 \mod 20$.
\end{itemize}
For $\alpha \in \{6,7,8,9\}$ there are only finitely many exceptions where the existence of Steiner systems is unknown. For this we refer the reader to the tables provided in section 3 of the book \cite{colbourn2010}.  \\ 
\hline
Transposed Steiner Systems & $(n,\theta,\frac{\theta-1}{\rho -1}, \rho)$ & ${\tiny \geq k\alpha - \binom{k}{2}}$  & ${\tiny 1 \leq k \leq \alpha}$ & File size equals ${\tiny k\alpha - \binom{k}{2}}$ if the original Steiner system has a maximal arc. \begin{itemize}\item There exist a Steiner system with $\alpha=3$ and a maximal arc if $\theta \equiv 3,7 \mod 12$.
\item  There exist a Steiner system with $\alpha=4$ and a maximal arc if $\displaystyle \frac{\theta-1}{3}$ is a prime power.  \end{itemize} To our best knowledge results about the existence of maximal arcs in Steiner systems with higher values of $\alpha$ are not known. \\
\hline
Grids & $(2a, a^2, a, 2)$ &  $ka - k^2/4$ for even $k$, $ka - (k^2 - 1)/4$ for odd $k$ & $1\leq k\leq a$ & 
The file size calculation can be done for any positive integer $a$.\\
\hline
MOLS (Remark \ref{MOLS1})  & $(4a, a^2, a, 4)$ & ${\tiny \geq k\alpha - \binom{k}{2}}$ & $1\leq k \leq 4$ & File size equals ${\tiny k\alpha - \binom{k}{2}}$ if $k=4$. Use the construction of two MOLS  for order greater than  6 \cite{bose1960}.  \\
\hline
MOLS (Lemma \ref{MOLS2}) & $ (\rho p^m, p^{2m}, p^m,$ \newline $\rho \leq p^m - 1)$ & ${\tiny kp^m - \binom{k}{2}}$ & $1 \leq k \leq \rho$ & $p$ is a prime. Use the construction of MOLS where the order is a prime power.\\
\hline
\end{tabular}
\end{center}
\caption{\label{table:d_greater_k_beta_1}Constructions where $d \geq k$ and $\beta = 1$. Note that we can perform trivial $\beta$-expansion to obtain higher $\beta$.  }
\end{table*}


\begin{table*}[t]
\small
\centering
\begin{tabular}{| p{1.5cm} |  p{3.5cm}  | p{5cm} |}
\hline
$\alpha$ & $\theta$  & Comments\\
\hline
\hline
$q$ & $q^2$  & $q$ is a prime. \\
\hline
$q+1$ & $q^2+q+1$ & $q$ is a prime. \\
\hline
$q+1$ & $q^3+1$ & $q$ is a prime. These designs are known as Unitals. \\
\hline
$2^r$ & $2^{r+s}+2^r-2^s$ & $2\leq r<s$. These designs are known as Denniston designs. \\
\hline
\end{tabular}
\vspace{2mm}
\caption{\label{Steiner-Systems}Well-known infinite families of Steiner systems when $t=2$. These can be found in \cite{colbourn2010}.}
\end{table*}

Many of our constructions will result from combinatorial designs that we briefly introduce (a detailed description can be found in \cite{stinson2004}).
\begin{definition}[\textbf{Combinatorial Design}]
A \textit{combinatorial design} (or, simply a design) is a pair $(\Omega, V)$ where $\Omega$ is a finite set of elements called ``points" and $V$ is a collection of non-empty subsets of $\Omega$ called ``blocks".
\end{definition}
A prototypical example with several applications is the balanced incomplete block design (BIBD).
\begin{definition}[\textbf{Balanced Incomplete Block Design}]
A $(\theta,\rho, \alpha,\lambda)$ \textit{balanced incomplete block design} (BIBD) is
a pair $(\Omega, V )$ that forms a combinatorial design such that $|\Omega| = \theta, |V| = n$; every element of $\Omega$ is contained in exactly $\rho$ blocks and every $2$-subset of $\Omega$ is contained in exactly $\lambda$ blocks.
\end{definition}
Let $n$ denote the number of blocks. By using combinatorial double counting arguments it can be seen that for a BIBD, the following relations hold.
\begin{align}
n \alpha &= \theta \rho,~\mbox{and} \label{bibd_eq_1}\\
\rho(\alpha - 1) &= \lambda(\theta -1). \label{bibd_eq_2}
\end{align}
A $(\theta,\rho, \alpha,\lambda)-$BIBD can be used as the FR code in a DSS as long as $\beta$-recoverability is guaranteed for an appropriate $\beta$ (there are several instances when $\beta = 1$). These $(\theta,\rho, \alpha,\lambda)-$BIBDs include finite projective planes and affine planes. Table \ref{Steiner-Systems} contains a list of well-known families of Steiner systems. A $(a^2+a+1,a+1, a+1,1)-$BIBD is equivalent to a \textit{projective plane} of order $a$. Projective planes have interesting geometric properties that can be used in determining the corresponding file size. For instance, any two blocks of a $(a^2+a+1,a+1, a+1,1)-$BIBD share exactly one point and any two points are contained in exactly one block in a projective plane. The smallest example of a projective plane corresponding to $a=2$ is known as the \textit{Fano plane} and is depicted in Fig. \ref{fanoplane1}. For more information on the projective planes and affine planes we refer the Chapter 2 of \cite{stinson2004}.


One can use the Fano plane to design the inner FR code, by interpreting the points as symbols and the blocks as storage nodes. Suppose we first apply a $(7,6)$-MDS to the file. Then we can place the coded symbols on the storage nodes as depicted in Fig. \ref{fanoplane1}. 
Note that these storage nodes are obtained from the blocks. The obtained DSS has the property that any two nodes share exactly one symbol. Thus, using Theorem \ref{thm:inc_enc} contacting any three nodes recovers at least $9 - \binom{3}{2} = 6$ distinct symbols and hence the file. Furthermore, we can identify a set of three nodes whose intersection is empty, e.g., nodes 1, 2 and 4. Thus, the maximum file size this DSS can support is 6. An affine plane can be obtained by deleting one block and its all points from a projective plane. Hence, an \textit{affine plane} of order $a$ is equivalent to a $(a^2,a+1, a,1)-$BIBD. Here any two points are contained in exactly one block. However, there are in general pairs of blocks that do not have any points in common. More generally, a FR code can be obtained from Steiner systems.

\begin{definition}[\textbf{Steiner Systems}] A $S(t, \alpha, \theta)$ \textit{Steiner system} is a set $\Omega$ of $\theta$ elements and a collection of subsets of $\Omega$ of size $\alpha$ called blocks such that any $t$-subset of the symbol set $\Omega$ appears exactly one of the blocks.
\end{definition}
Steiner systems are examples of $t$-designs. A FR code is a \textit{$t$-design} if every $t$-subset of symbols is contained in exactly $\lambda$ nodes. The concept of $t$-designs can be viewed as a generalization of the concept of BIBDs. Naturally, a $S(2, \alpha, \theta)$ Steiner system is a $(\theta,\rho,\alpha,1)$-BIBD where
$$ \rho=\frac{\theta-1}{\alpha-1}, ~~~n= \frac{(\theta-1)\theta}{(\alpha-1)\alpha},~~~d=\alpha~~\mbox{and}~~\beta=1. $$
Thus, projective planes and affine planes are instances of Steiner systems. A given FR code can be put in one-to-one correspondence with an incidence matrix as explained below.
\begin{definition}[\textbf{Incidence Matrix of a FR Code}] An \textit{incidence matrix of a FR code} $\mathcal{C}=(\Omega,V)$ where $\Omega = [\theta]$ and $V = \{V_1, V_2, \dots, V_n\}$ is the $\theta\times n$  binary matrix $N$ defined by
$$N_{i,j}=\left\{
    \begin{array}{ll}
      1, & \hbox{if $i \in V_j$;} \\
      0, & \hbox{otherwise.}
    \end{array}
  \right.$$
\end{definition}
We shall sometimes refer to the FR code $\mathcal{C}$ by simply referring to its incidence matrix $N$. We will occasionally refer to the bipartite graph corresponding to the FR code as well. This is defined next.
\begin{definition}[\textbf{Bipartite graph of a FR Code}] \label{def:bipartite_gr} For a FR code $\mathcal{C}=(\Omega,V)$ where $\Omega = [\theta]$ and $V = \{V_1, V_2, \dots, V_n\}$ with incidence matrix $N$, we define its bipartite graph $G_b = (V_l \cup V_r, E)$ as follows. We associate the storage nodes in $V$ with the vertices $V_l$ and the points in $\Omega$ with the vertices $V_r$ so that $V_l$ and $V_r$ are disjoint. There exists an edge between $v \in V_l$ and and $u \in V_r$ if and only if $N(u,v) = 1$.
\end{definition}
\begin{table*}[t]
\small
\centering
\begin{tabular}{| p{1.5cm} |  p{3.5cm} |  p{1.0cm} | p{3.5cm} | p{1.5cm} | p{5.0cm} |}
\hline
Method & $(n,\theta, \alpha, \rho)$ & $\beta$ & $\mathcal{M}$ & Range of $k$ & Comments\\
\hline
\hline
Affine Resolvable Designs & $(q \rho, q^m, q^{m-1}, 1 \leq \rho \leq \frac{q^m-1}{q-1})$ & $q^{m-2}$ & $q^m\bigg{(} 1 - \bigg{(}1 - \frac{1}{q}\bigg{)}^k\bigg{)}$ & $1 \leq k \leq m$ & The file size exceeds the trivial lower bound $k\alpha-\beta \binom{k}{2}$. The parallel classes need to be chosen in a careful manner. If $q > m$, then we choose $\rho > m$ and if $q \leq m$, we choose $\rho \leq m$. \\ 
\hline
Hadamard Designs & $(8a-2,4a,2a,4a-1)$ & $a$ & $3a$ & $1 \leq k \leq 2$ & $\beta$ is not restricted to be a prime power. \\
\hline
\end{tabular}
\vspace{2mm}
\caption{\label{table:d_greater_k_beta_larger_1}Constructions where $d \geq k$ and $\beta > 1$. In many cases these construction parameters cannot be obtained by trivial $\beta$-expansion.}
\end{table*}

\begin{table*}[t]
\small
\centering
\begin{tabular}{| p{4.5cm} | p{2.0cm} |  p{2.0cm} |  p{0.75cm} | p{1.5cm} | p{1.5cm} | p{3.0cm} |}
\hline
Base Code & Method & $(n,\theta, \alpha, \rho)$ & $\beta$ & $\mathcal{M}$ & Range of $k$ & Comments\\
\hline
\hline
 $\mathcal{C}^T$ obtained via the transpose of a Steiner system $S(2,\tilde{\alpha},\tilde{\theta})$ with $\tilde{\rho} = \frac{\tilde{\theta}-1}{\tilde{\alpha} -1}$ with maximal arc of size $\tilde{\rho}+1$ &  Kronecker product of $\mathcal{C}^T$ with itself &  $(\tilde{n}^2, \tilde{\theta}^2, \tilde{\alpha}^2, \tilde{\rho}^2)$ & $\tilde{\alpha}$ &  $k\tilde{\rho}^2 - \tilde{\rho} \binom{k}{2}$ & $1 \leq k \leq \tilde{\rho}$ & There exist a Steiner system with $\tilde{\alpha}=3$ and a maximal arc if $\tilde{\theta} \equiv 3,7 \mod 12$. In several cases, the codes obtained via Kronecker product cannot be obtained by trivial $\beta$-expansion.\\
\hline
\end{tabular}
\vspace{2mm}
\caption{\label{table:kronecker} Constructions obtained via Kronecker product, where $d \geq k$.}
\end{table*}

\begin{table*}[t]
\small
\centering
\begin{tabular}{| p{3.5cm} | p{2.0cm} |  p{0.5cm} | p{1.5cm} | p{1.7cm} | p{5.0cm} |}
\hline
 Method & $(n,\theta, \alpha, \rho)$ & $\beta$ & $\mathcal{M}$ & Range of $k$ & Comments\\
\hline
\hline
Use undirected graph $\Gamma = (V,E)$, $|V| = n$, degree $s$ and girth $g$ & $(n,\frac{ns}{2},s,2)$ & 1 & $k(s-1)$ & $k = as+b$ & Need $s > b \geq a+1$ and $k \leq g$. Codes meet the minimum distance bound for locally recoverable codes. Since an $(s,g)$-cage minimizes the number nodes in the system, we will have highest possible code rate for this particular construction. Here we list some of the well-known infinite families of $(s,g)$-cage. \begin{itemize}
\item   $(s,3)$-cages are complete graphs on $s+1$ vertices.
\item   $(s,4)$-cages are complete bipartite graphs on $2s$ vertices.
\item   When $s-1$ is a prime power $(s,6)$-cages can be obtained from incidence graphs of projective planes.
\item   When $s-1$ is a prime power $(s,8)$-cages and $(s,12)$-cages can be obtained from incidence graphs of generalized polygons.
\end{itemize} For more see the survey \cite{Exoo08}.\\
\hline
Use $l$ copies of FR code $\mathcal{C} = (\Omega,V)$ with parameters $(\tilde{n},\tilde{\theta},\tilde{\alpha},\tilde{\rho})$ such that any $\tilde{\Delta}+1$ nodes cover $\tilde{\theta}$ symbols. $|V_i \cap V_j| \leq \tilde{\beta}$ for $i \neq j$. Parameters satisfy $(\tilde{\rho}-1)\tilde{\alpha}\tilde{\theta} - (\tilde{\theta}+\tilde{\alpha})(\tilde{\Delta}-1)\tilde{\beta} \geq 0$. & $(l\tilde{n},l\tilde{\theta},\tilde{\alpha},\tilde{\rho})$ & $\tilde{\beta}$ & $t\tilde{\theta}+\tilde{\alpha}$ & $t\tilde{n}+1$ & Codes meet the minimum distance bound for locally recoverable codes with exact and uncoded repair ({\it cf.} Section \ref{sec:locally_recoverable_dss}).\\
\hline
\end{tabular}
\vspace{2mm}
\caption{\label{table:local_fr} Constructions of local FR codes where $d < k$.}
\end{table*}

\begin{definition}[\textbf{Transposed FR Code}]
For a FR code $C$ with incidence matrix $N$, the code specified by $N^T$ is called \textit{transposed FR code} of $C$ and denoted by $C^T$ if the design obtained from $N^T$ is $\beta$-recoverable for some $\beta$. 
\end{definition}
Note that, in the transposed code, the roles of the storage nodes and the symbols are reversed. An infinite family of transposed codes can be obtained from Steiner systems with $t=2$. In such Steiner systems any pair of symbols is contained in exactly one node which implies that any pair of nodes in the transposed design share exactly one symbol. This in turn means that the transposed design is 1-recoverable. 

Incidence matrices with appropriate parameters can be combined via operations such as the Kronecker product to obtain new matrices (equivalently FR codes) with a new set of parameters. We use this technique extensively in the sequel to generate families of FR codes.
\begin{definition}[\textbf{Kronecker Product}]
If $A$ is an $m$-by-$r$ matrix and $B$ is a $p$-by-$q$ matrix, then the \textit{Kronecker product} $A\otimes B$ is the $mp$-by-$rq$ matrix
\[
\left(
\begin{array}{cccc}
 a_{11}B &a_{12}B &\cdots&a_{1r}B  \\
 a_{21}B &a_{22}B &\cdots&a_{2r}B  \\
 \vdots&\vdots&\vdots&\vdots\\
   a_{m1}B &a_{m2}B &\cdots&a_{mr}B
\end{array}
\right).
\]
 \end{definition}

Let $N_1$ and $N_2$ be two incidence matrices of  FR codes  $\mathcal{C}_1=(\Omega_1,V_1)$ and $\mathcal{C}_2=(\Omega_2,V_2)$ with parameters $(n_1, \theta_1, \alpha_1, \rho_1)$ and  $(n_2, \theta_2, \alpha_2, \rho_2)$ respectively. Let $c_1, \cdots, c_{n_1}$  be the $n_1$ columns of $N_1$ and $d_1, \cdots, d_{n_2}$  be the $n_2$ columns $N_2$ . A new FR code can be obtained from the old one by the following incidence matrix
$$\bar{N}=
\left[
\begin{matrix}
 N_1 \otimes d_1& N_1 \otimes d_2&\cdots&N_1 \otimes d_{n_2}
\end{matrix}
\right].$$

We can find an appropriate permutation matrix $P$ such that  the matrix $\bar{N}P$ is equal to the Kronecker product of $N_1$ and $N_2$. Note that  matrix $P$ reorders the columns of $\bar{N}$.


We can obtain a DSS by replicating the symbols of another DSS via the Kronecker product. In the subsequent discussion we will refer to this technique for obtaining codes with $\beta > 1$ as \textit{trivial $\beta$-expansion}.
\begin{definition} \label{def:trivial_beta} [\textbf{Trivial $\beta$-expansion}] Let $N$ be incidence matrix of a FR code $C$ with parameters $(n, \theta, \alpha, \rho)$ with $\beta=1$. Let $\mathds{1}$ be the $m \times 1$ all-ones column vector. The FR code $\hat{C}$ obtained from $\hat{N}=N \otimes \mathds{1}$ which has parameters $(n, \theta m, \alpha m, \rho)$ is called a trivial $\beta$-expansion of the code $C$ with $\beta=m$.
\end{definition}
In the remainder of this section, we discuss some illustrative examples of FR codes. Our first example is a code with $\beta > 1$ that cannot be obtained by trivial $\beta$-expansion.
\begin{example}[\textbf{A Non-trivial Code with $\beta > 1$}]
  \label{eg:non_trivial_beta}
 Consider the DSS shown in Fig. \ref{fig:non_trivial_large_beta}. The ten symbols are obtained by using an outer $(10,6)$ MDS code followed by the FR code illustrated in Fig. \ref{fig:non_trivial_large_beta}. Note that the DSS can recover from a single node failure by downloading two packets each from two nodes in the same column; hence $d=2$. 
 Moreover, any two nodes share 0,1, or 2 symbols in common which implies that any two nodes recover at least $6$ symbols, thus $k=2$. According to the Singleton bound $d_{\min} \leq 15 - \lceil \frac{6}{4} \rceil + 1 = 14$. 
The system requires only two surviving nodes to recover the file thus the code is resilient up to 13 failures (since $k=2$) and thus meets the Singleton bound.  However, this code (with $\beta = 2$) cannot be arrived at simply by replication. To see this we note that if this were true, the original DSS with $\beta = 1$ must correspond to a storage capacity of $2$ and have a number of symbols which is 5. However, this means that there can be at most $\binom{5}{2} = 10$ distinct storage nodes of capacity two. Thus our design with $n = 15$ cannot be obtained this way. 
\end{example}

The idea underlying Example \ref{eg:non_trivial_beta} can be formalized as follows.
\begin{observation}[\textbf{Non-trivial FR Codes with $\beta>1$}] \label{non_trivial_FR} A FR code with parameters $(n, \theta m, \alpha m, \rho)$, $\beta=m$ and distinct storage nodes cannot be obtained from a trivial $\beta$-expansion if $n>\binom{ \theta }{\alpha}$.
\end{observation}

Next, we demonstrate an example of a locally recoverable DSS, i.e., a system where $d < k$ that is constructed using the Kronecker product method. 

\begin{example}[\textbf{Locally Recoverable Code Using Kronecker Product Technique}] Let $\mathcal{C}=(\Omega,V)$ be a FR code with $\Omega=\{1,2,3\}$ and  $V=\{V_1=\{1,2\},V_2=\{2,3\}, V_3=\{1,3\}\}$ with incidence matrix $N$. The code obtained from $\bar{N}=I \otimes N$ is presented in Fig. \ref{fig:kron_local_eg} where $I$ denotes the $3 \times 3$ identity matrix. Suppose that the outer MDS code has parameters $(9,5)$, so that $\theta = 9, \calM = 5$. Consider contacting any of the four nodes depicted in Fig. \ref{fig:kron_local_eg}. These nodes will fall into one of the three columns in the figure. So, there are three cases we need to examine.
\begin{itemize}
\item \textbf{Case (a):} Two nodes can be chosen from one of the columns and one from each of the rest. The union of these nodes has a cardinality of $7$.
\item \textbf{Case (b):} We first select two columns and two nodes within each column. In this case the size of the union is $6$.
\item \textbf{Case (c):} Finally, we can select two columns and choose three nodes in one column and one node in the other column. In this case the cardinality of the union is $5$.
\end{itemize}
Thus, it is evident that contacting any $k=4$ nodes will recover at least $5$ symbols. Note that a failed node can be recovered by contacting the remaining two nodes in its column by downloading one packet from each of them. Thus, $d = 2 < k$. This implies that the code is locally recoverable. 
By applying a similar case analysis for the failure patterns we can conclude that the code is resilient to 5 failures and it  meets the minimum distance bound in Lemma \ref{local_bound}.

\end{example}

\subsection{Summary of Contributions}
In this work we present several constructions of FR codes. The contributions of our work can be summarized as follows.
We construct a large class of FR codes for $d \geq k$ from combinatorial structures such as grids, mutually orthogonal Latin squares (MOLS), resolvable designs and Hadamard designs. These were first presented in the literature in the conference version of the current manuscript \cite{olmez2012}. While \cite{el2010} presented constructions based on Steiner systems, our work presents a rigorous analysis of the file size of the corresponding DSS. The Kronecker product technique for generating new DSS from existing ones is also new \cite{olmez2013_2}. Furthermore, our conference paper \cite{olmez2013} was the first to present locally recoverable FR codes where $d < k$.

Tables \ref{table:d_greater_k_beta_1} -- \ref{table:local_fr} contain a description of the various constructions and the corresponding DSS parameter values that can be achieved by these constructions. We defer an in-depth discussion of these parameters to the respective sections. However, we highlight the key contributions of our work by referring to appropriate rows of Tables \ref{table:d_greater_k_beta_1} -- \ref{table:local_fr} below. Specific details about the construction techniques can be found in the corresponding sections of the paper.
\begin{itemize}
\item We construct a large class of FR codes based on resolvable designs \cite{stinson2004} where the repetition degree ($\rho$) of the symbols can be varied in an easy manner (see Table \ref{table:d_greater_k_beta_1} (rows 3 -- 5) and Table \ref{table:d_greater_k_beta_larger_1}). The constructions of \cite{el2010} lack this flexibility as they are mostly based on Steiner systems where the repetition degree is usually fixed by the construction.
\item We construct FR codes where $\beta > 1$, i.e., the new node downloads more than one packet from the $d$ surviving nodes. We emphasize that starting with a FR code with $\beta = 1$, it is trivially possible to arrive at a code with $\beta > 1$ by trivial $\beta$-expansion ({\it cf.} Definition \ref{def:trivial_beta}). However, such a strategy only results in a limited range of system parameters that can be achieved. We present several codes (see Tables \ref{table:d_greater_k_beta_larger_1} and \ref{table:kronecker}) that achieve certain parameter ranges that cannot be achieved in a trivial manner.
\item Determining the file size that can be supported by a given FR code turns out be challenging. 
Much of the literature in combinatorial designs only discusses the pairwise overlaps between the content of the different storage nodes. However, the file size depends on the union of all subsets of storage nodes of size $k$. In this work we determine the file sizes for most of our constructions. In particular, we demonstrate a family of FR codes whose file size is strictly larger than a simple lower bound that is obtained by applying the inclusion-exclusion principle (see row 1, Table \ref{table:d_greater_k_beta_larger_1}). We also determine the file size for a large class of codes obtained from Steiner systems that were originally considered in \cite{el2010} (see row 2, Table \ref{table:d_greater_k_beta_1}). Several of our constructions are shown to meet the Singleton bound for specific file sizes, which demonstrates their optimality.
\item We present the Kronecker product as a technique for constructing new FR codes from existing ones (Table \ref{table:kronecker}) and analyze the properties of codes thus obtained.
\item In this work, we propose a large family of {\it locally recoverable} FR codes where $d < k$, i.e., the repair degree is strictly smaller than the number of nodes contacted for recovering the stored file. We derive an appropriate minimum distance bound for our class of codes that enjoy local, exact and uncoded repair, and demonstrate constructions that meet these bounds (Table \ref{table:local_fr}).
\end{itemize}

\subsection{Discussion of related work}
The work of Dimakis et al. \cite{dimakis2010} initiated the work on regenerating codes, by demonstrating the tradeoff between the storage capacity of nodes and the repair bandwidth. Their work considered functional repair, where the new node is functionally equivalent to the failed node and demonstrated that random network coding suffices for achieving this tradeoff. Following this, several papers \cite{rashmi2011, rashmi2009, SuhR11, tian2013, papailiopoulos2012, shah2012, tamo2011, el2010, olmez2012} considered the construction of exact repair regenerating codes, where the new node is an exact copy of the failed node. In most cases, these constructions either operate at the minimum storage regenerating (MSR) point \cite{rashmi2011,shah2012,papailiopoulos2011,tamo2011,SuhR11}  or the minimum bandwidth regenerating (MBR) point \cite{rashmi2011, rashmi2009, shah_et_al12_rep_by_transfer, el2010, olmez2012}. More recently, codes with local repair have been investigated where the metric for repair is the number of surviving nodes that are contacted for repair \cite{gopalan2012,oggier2011,papD12,kamath_et_al_14,rawatKSV14,olmez2013}.

Constructions of repair-by-transfer codes, where node repair is performed simply by downloading symbols from surviving nodes was first presented in the work of \cite{shah_et_al12_rep_by_transfer} where they constructed a repair-by-transfer MBR code with $d = n-1$. Repair by transfer codes have also appeared in \cite{ShumH12,HuLS13}. The work of \cite{el2010} also considered such codes (termed ``exact and uncoded repair") but with a repair degree that can be strictly smaller than $n-1$. The repair operates by contacting a specific set of $d$ surviving nodes and is hence table based. Reference \cite{el2010} introduced the system architecture whereby an MDS code is applied to a file consisting of $\calM$ symbols to obtain $\theta$ symbols. These symbols are then placed onto the storage nodes and this placement is referred to as the fractional repetition (FR) code. The codes in \cite{el2010}, were derived from Steiner systems. They provided lower and upper bounds on the corresponding file sizes. Following this, the work of \cite{koo2011} constructed FR codes from bipartite cages. These codes enjoy the property that the node storage capacity is much larger than the replication degree. For the given parameters they design codes with the smallest number of storage nodes. In \cite{koo2011}, they used MOLS to construct bipartite cages and the codes thus obtained are different from ours. In our construction we obtain the storage nodes directly from the set of MOLS and also obtain net FR codes. Reference \cite{ernvall12} presents necessary and sufficient conditions on the existence of a FR code with certain parameters; however, it does not consider the issue of determining the file size for a given $k$. 

The work of \cite{silberstein2014} presents several FR code constructions based on combinatorial structures including regular and biregular graphs, graphs with a given girth, transversal designs, projective planes and generalized polygons. They consider codes where $\alpha = d \geq k$ and $\beta = 1$ and show that the file size of their constructions meets the upper bound presented in \cite{el2010} for $k \leq d$. This work is closely related to the content of Section \ref{sec:resolv_design_dss} of our work. Their construction of FR codes from transversal designs treats the blocks of the transversal design as symbols. Thus, it can be considered as working with the transpose of the incidence matrix corresponding to the original transversal design. Our FR codes in Section \ref{sec:resolv_design_dss} are obtained from nets which can also be viewed as transposes of transversal designs. However, as discussed in Section \ref{sec:discuss_code_params}, the analysis of file size for our constructions cannot be obtained from the results in \cite{silberstein2014}. Our work differs in the sense that we present constructions with non-trivial $\beta$ values, Kronecker product constructions and local FR codes.

The problem of local repair for scalar codes ($\alpha=1$) was first considered in \cite{gopalan2012}. This was extended to vector codes ($\alpha>1$) in \cite{kamath_et_al_14,papD12}. References \cite{kamath_et_al_14,papD12} study the tradeoff between locality and minimum distance and corresponding code constructions. In \cite{kamath_et_al_14}, the authors presented constructions that use the repair-by-transfer MBR codes of \cite{shah_et_al12_rep_by_transfer} as individual components. Local codes were also studied in \cite{kamath2013} where the design consists of an outer Gabidulin encoder followed by inner local MBR encoders. This work (see Construction III.1 in \cite{kamath2013}) also provides examples of local FR codes by using $t$-designs. However, the achievable parameters are limited as $k$ needs to be chosen to be at most $t$ and explicit constructions of $t$-designs for large $t$ are largely unknown (when $t \geq 3$ there are only finitely many known explicit constructions \cite{colbourn2010}). In Section \ref{sec:locally_recoverable_dss} we focus on regenerating codes that allow a repair process in a local manner by simply downloading packets from the surviving nodes. We provide an upper bound for the minimum distance and constructions of codes which meet this bound. Our constructions use local FR codes instead of repair by transfer MBR codes. 
We also note that our codes are quite different from those that appear in \cite{kamath2013, kamath_et_al_14} and allow for a larger range of code parameters. Regenerating codes using $t$-designs were also presented in \cite{tian2013}. The architecture of the codes consists of a layered erasure correction structure that ensures a simple decoding process. These codes are showed to be achieve performance better than time-sharing between MBR and MSR points.

\section{Construction of FR codes when $k \leq d$}
\label{sec:resolv_design_dss}

In this section we present the construction of FR codes where $d \geq k$. As discussed in Example \ref{eg:beta_rec_cond_fr} it is possible that certain set systems do not satisfy the property of $\beta$-recoverability and hence cannot be used to construct FR codes. However, there are a large class of combinatorial designs that can be used to construct FR codes. In particular, we present various constructions of FR codes that are derived from balanced incomplete block designs (BIBDs) and resolvable designs. Our constructions address several issues that exist with prior constructions in the literature. For instance, resolvable designs allow the repetition degree of the symbols in the FR code to be varied in a simple manner, a flexibility that prior constructions typically lack. We present a large class of codes that cannot be obtained via trivial $\beta$-expansion.

Our first set of constructions are FR codes based on Steiner systems with $t=2$ (that are BIBDs) which have been previously considered in the literature \cite{el2010}. However, to our best knowledge, prior work does not provide results on the file size of the constructions. In the discussion below, we present a certain class of Steiner systems for which we can determine the file size of the FR codes obtained from their transpose. To demonstrate the difficulty of determining the file size for a general Steiner system, we first discuss two non-isomorphic Steiner systems with the same parameter values that result in FR codes with different file sizes. This demonstrates that file size calculations for Steiner systems cannot be performed just based on the system parameters. Accordingly, we consider Steiner systems that have maximal arcs \cite{quattrocchi1993,greig2003}. It turns out that we can determine the file size of the corresponding transposed codes.


\subsection{FR codes from Steiner systems}
We consider Steiner systems $S(2,\alpha, \theta)$. Note that the repetition degree of any symbol is $\rho=\frac{\theta-1}{\alpha-1}$ and any two distinct symbols are contained in exactly one node. Consider the FR code $\mathcal{C}=(\Omega,V)$ obtained from it and its transpose.

In general, it is a challenging task to find the file size for a given FR code. For codes obtained from Steiner systems and their transposes, lower bounds based on the inclusion-exclusion principle were presented in \cite{el2010}. However, it is important to note that the file size depends critically on the structure of the Steiner system, i.e., two Steiner systems with the same parameters can have different file sizes. To see this, consider two non-isomorphic Steiner systems $S(2,3,15)$ denoted $\mathcal{D}_1$ and $\mathcal{D}_2$; the nodes of these designs are provided in Tables \ref{Design1} and \ref{Design2}. These designs can also be found in \cite{colbourn2010}.

Let $S$ be a subset of symbols of the design such that no $3$-subset of $S$ is contained in a node. 
By checking all subsets of the symbol set one can observe that the maximum size of $S$ in $\mathcal{D}_1$ and $\mathcal{D}_2$ equals $6$ ($\{0,1,3,6,7,9\}$) and $8$ ($\{1, 2, 4, 6, 7, 8, 9, 13\}$) respectively. 

This observation results in different file sizes in the codes obtained from the transposes of $\mathcal{D}_1$ and $\mathcal{D}_2$, denoted $\mathcal{D}_1^T$ and $\mathcal{D}_2^T$ respectively. In fact for $k=7$, the design $\mathcal{D}_2^T$ yields a code which has file size $\mathcal{M}_2=28$ which matches the inclusion-exclusion lower bound given by $7 \times 7 - \binom{7}{2}$. However, the design $\mathcal{D}_1^T$ yields a code with file size $\mathcal{M}_1=29$ which is strictly larger\footnote{This example corrects an error in Lemma 11 of \cite{el2010}.}.

We now elaborate on the role of $S$ in the above example. Firstly, note that if $\mathcal{D}_i$ is a Steiner system, then any two storage nodes in $\mathcal{D}_i^T$ intersect in one symbol. Consider the corresponding transposed codes $\mathcal{D}_1^T$ and $\mathcal{D}_2^T$, where the roles of symbols and nodes is now reversed. As $S$ for $\mathcal{D}_2$ is of size $8$, it implies that we can pick $k=7$ storage nodes in $\mathcal{D}_2^T$ such that the intersection of any three storage nodes is empty (owing to the definition of $S$). Thus, upon applying the inclusion-exclusion principle, we obtain the file size to be $7 \times 7 - \binom{7}{2} = 28$.

In contrast, the maximum size of $S$ in $\mathcal{D}_1$ is $6$. Thus, for any set of $k=7$ storage nodes in $\mathcal{D}_1^T$ there is at least one three-way intersection that is non-empty. Upon exhaustive enumeration, one can realize that the file size in this case is $29$ which is strictly higher than $28$.

The notion of the set $S$ introduced above can be formalized in terms of a maximal arc in Steiner systems. For Steiner systems that possess a maximal arc, we can therefore determine the file size. In addition, prior results in \cite{quattrocchi1993,greig2003}, demonstrate that such maximal arcs exist in a large class of Steiner systems. In the discussion below, we make these arguments in a formal manner.

\begin{table}[t]
\begin{center}
\begin{tabular}{|c| c | c | c | c |}
  \hline
 $\left\{0, 1, 2\right\}$ & $\left\{0, 3, 4\right\}$ & $\left\{0, 5, 6\right\}$& $\left\{0, 8, 7\right\}$ &$\left\{0, 9, 10\right\}$\\
  \hline
 $\left\{0, 11, 12\right\}$ & $\left\{0, 13, 14\right\}$ &$\left\{1, 3, 5\right\}$ & $\left\{1, 4, 7\right\}$ & $\left\{8, 1, 6\right\}$\\
 \hline
$\left\{1, 11, 9\right\}$ & $\left\{1, 10, 13\right\}$ & $ \left\{1, 12, 14\right\}$ & $\left\{9, 2, 3\right\}$ & $\left\{2, 4, 6\right\}$\\
 \hline
$ \left\{2, 10, 5\right\}$& $\left\{2, 14, 7\right\}$ & $\left\{8, 2, 12\right\}$ & $\left\{2, 11, 13\right\}$ & $\left\{3, 11, 6\right\}$\\
 \hline
$\left\{3, 12, 7\right\}$ & $\left\{8, 3, 13\right\}$ & $\left\{10, 3, 14\right\}$ & $\left\{4, 5, 13\right\}$ &$\left\{8, 9, 4\right\}$\\
 \hline
  $\left\{4, 10, 12\right\}$ &$\left\{11, 4, 14\right\}$ &$\left\{11, 5, 7\right\}$ & $\left\{8, 5, 14\right\}$ & $\left\{9, 12, 5\right\}$\\
 \hline
$\left\{10, 6, 7\right\}$ & $\left\{9, 6, 14\right\}$ & $\left\{12, 13, 6\right\}$ & $\left\{9, 13, 7\right\}$ & $\left\{8, 10, 11\right\}$\\
 \hline
\end{tabular}
\end{center}
\caption{\label{Design1} Nodes of the Steiner system $\mathcal{D}_1$}
\end{table}

\begin{table}[t]
\begin{center}
\begin{tabular}{|c| c | c | c | c |}
  \hline
 $\left\{1, 11, 6\right\}$ &$\left\{1, 2, 5\right\}$ & $\left\{2, 3, 6\right\}$ & $\left\{9, 5, 6\right\}$ &$ \left\{3, 11, 5\right\}$\\
  \hline
 $\left\{7, 13, 5\right\}$& $\left\{11, 4, 13\right\}$ & $\left\{2, 12, 7\right\}$ & $\left\{3, 4, 7\right\}$ & $\left\{8, 4, 5\right\}$ \\
 \hline
$\left\{8, 11, 7\right\}$ & $\left\{4, 12, 6\right\}$ & $\left\{8, 6, 14\right\}$ & $\left\{12, 5, 14\right\}$ & $\left\{8, 3, 13\right\}$\\
 \hline
 $\left\{8, 9, 12\right\}$ & $\left\{9, 10, 13\right\}$ &$\left\{1, 12, 13\right\}$ &$\left\{0, 9, 7\right\}$ & $ \left\{9, 2, 11\right\}$\\
 \hline
  $\left\{0, 8, 2\right\}$ &$\left\{9, 4, 14\right\}$ &$\left\{10, 11, 14\right\}$ & $\left\{0, 11, 12\right\}$ & $\left\{0, 3, 14\right\}$\\
 \hline
 $\left\{8, 1, 10\right\}$ & $\left\{10, 3, 12\right\}$ & $\left\{1, 3, 9\right\}$ & $\left\{0, 10, 5\right\}$ &$\left\{0, 13, 6\right\}$ \\
 \hline
$\left\{1, 14, 7\right\}$&$\left\{2, 4, 10\right\}$& $\left\{2, 13, 14\right\}$ & $\left\{0, 1, 4\right\}$ & $\left\{10, 6, 7\right\}$\\
 \hline
\end{tabular}
\end{center}
\caption{\label{Design2} Nodes of the Steiner system $\mathcal{D}_2$ }
\end{table}

\begin{definition}[\textbf{$s$-arc}]  Let $(\Omega, V)$ be a design. A subset $S \subset \Omega$ with $|S| = s$ is called an \textit{$s$-arc} if for each node $V_i \in V$   either $|V_i \cap S| = 0$ or $|V_i \cap S| = 2$ holds. \end{definition}

The definition of $s$-arc implies that any three symbols from $S$ are not contained in any node in $V$. The largest set $S$ with this property is called a \textit{maximal arc} of the design \cite{assmus1992}. It turns out that we can determine the file size for FR codes obtained from transposes of Steiner systems with nontrivial maximal arcs.
\begin{table*}[t]
\begin{center}
\begin{tabular}{|c| c | c | c |}
  \hline
 $\{1, 4, 7, 8\}$ & $\{0, 2, 8, 9\}$ & $\{1, 3, 5, 9\}$ & $\{2, 4, 5, 6\}$\\
  \hline
$\{0, 3, 6, 7\}$ & $\{6, 9, 12, 13\}$ & $\{5, 7, 13, 14\}$ & $\{6, 8, 10, 14\}$\\
 \hline
$\{7, 9, 10, 11\}$ & $\{5, 8,11, 12\}$ & $\{2, 3, 11, 14\}$ & $\{3, 4, 10, 12\}$\\
 \hline
$ \{0, 4, 11, 13\}$ & $\{0, 1, 12, 14\}$ & $\{1, 2, 10, 13\}$ & $\{0, 5, 10, 15\}$\\
 \hline
$\{1, 6, 11, 15\}$ & $\{2, 7, 12, 15\}$ & $\{3, 8,13, 15\}$ & $\{4, 9, 14, 15\}$\\
 \hline
\end{tabular}
\end{center}
\caption{\label{Design3} Nodes of the Steiner system $S(2, 4,16)$}
\end{table*}
\begin{table*}[t]
\begin{center}
\begin{tabular}{|c| c | c | c |}
  \hline
 $\{1, 4, 12, 13, 15\}$ & $\{0, 2, 13, 14, 16\}$ & $\{1, 3, 10, 14, 17\}$ & $\{2, 4, 10, 11, 18\}$\\
  \hline
$ \{0, 3, 11, 12, 19\}$ & $\{2, 3, 6, 9, 15\}$ & $\{3, 4, 5, 7, 16\}$ & $\{0, 4,6, 8, 17\}$\\
 \hline
$\{0, 1, 7, 9, 18\}$ &  $\{1, 2, 5, 8, 19\}$ & $\{7, 8, 11, 14, 15\}$ & $\{8,9, 10, 12, 16\}$\\
 \hline
$ \{5, 9, 11, 13, 17\}$ & $\{5, 6, 12, 14, 18\}$& $\{6, 7, 10, 13,19\}$ &$\{15, 16, 17, 18, 19\}$\\
 \hline
\end{tabular}
\end{center}
\caption{\label{Design4} Transposed code obtained from the Steiner system $S(2,4,16)$}
\end{table*}

For a maximal arc $S$, consider a symbol $q \in S$. In this case there are $s-1$ pairs of symbols $(p,q)$ such that $p,q \in S$. Since $\mathcal{C}$ is a Steiner system and $S$ is a maximal arc there are $s-1$ distinct nodes in $V$ where each of these pairs occurs. Now, the repetition degree of the system is $\rho$. Thus, there are $\rho - (s-1)$ nodes which contain the symbol $q$ but no other symbol from $S$. Based on our assumption, each node in $V_i \in V$ is such that either $|V_i \cap S| = 0$ or $|V_i \cap S| = 2$. Thus, it has to be the case that $s = \rho + 1$.

\begin{lemma}\label{lemma:maximal_arc}
Let $\mathcal{C} = (\Omega,V)$ be a FR code derived from a Steiner system $S(2, \alpha, \theta)$ with $\rho=\frac{\theta-1}{\alpha-1}$, such that it has a maximal arc of size $\rho + 1$. Then, the transposed FR code $\mathcal{C}^T$ is such that its code rate is $\frac{k \rho - \binom{k}{2}}{n\alpha}$ for $1\leq k \leq \rho + 1$.
\end{lemma}
\begin{proof}
In the transposed code $\mathcal{C}^T$, consider any subset of nodes of size $k$, where $1 \leq k \leq \rho + 1$. As any two symbols in the original code $\mathcal{C}$ occur in exactly one node of $\mathcal{C}$ it holds that two nodes $V_1$ and $V_2$ in $\mathcal{C}^T$ are such that $|V_1 \cap V_2| = 1$. In addition, the storage capacity of the nodes in $\mathcal{C}^T$ is equal to $\rho$.

Using the inclusion-exclusion principle ({\it cf.} Theorem \ref{thm:inc_enc}), we observe that these nodes cover at least $k\rho - \binom{k}{2}$ symbols in $\mathcal{C}^T$. Now we pick a set of $k$ nodes in $\mathcal{C}^T$ that correspond to a subset of the maximal arc $S$ in $\mathcal{C}$. Based on the argument above, it is clear that any two of these nodes intersect in exactly one symbol and any $l$ of the nodes have an empty intersection if $l \geq 3$. It follows that the union of these nodes has exactly $k \rho - \binom{k}{2}$ symbols. The result follows.
\end{proof}

Next we provide an explicit  example. Let $\mathcal{C}$ be the FR code obtained from a Steiner system $S(2,\alpha = 4, \theta = 16)$.

\begin{example}[\textbf{File size of FR code obtained from the transpose of Steiner System $S(2,\alpha = 4, \theta = 16)$}]
The nodes in $\mathcal{C}$ are specified in Table \ref{Design3} and the nodes of the transposed code $\mathcal{C}^T$ are specified in Table \ref{Design4}.


Since the maximal arc should be a set of with cardinality $6$, we can choose the symbols greedily and construct the set $S=\{0,1,2,3,4,15\}$ as a maximal arc for this Steiner system

According to Lemma \ref{lemma:maximal_arc}, the file size for $\mathcal{C}^T$ for $1\leq k \leq 6$  can be determined by just considering the nodes
\begin{align*}
&\{1, 4, 12, 13, 15\}, \{0, 2, 13, 14, 16\}, \{1, 3, 10, 14, 17\}, \\
&\{2, 4, 10, 11, 18\}, \{0, 3, 11, 12, 19\},~\mbox{and}~\{15, 16, 17, 18, 19\}
\end{align*}
as these correspond to the symbols of $S$ in $\mathcal{C}$.
For these values of $k$, the file size of the code is  $5k-\binom{k}{2}$. Moreover, it is optimal with respect to Singleton bound for $1\leq k \leq 3$ ({\it cf.} Observation \ref{obs:meet_bound}). 
\end{example}

\begin{remark}[\textbf{Steiner Systems with $\alpha=3,4$}]
It is known that several Steiner systems possess maximal arcs. Here we provide the known results for small values of $\alpha$.
\begin{itemize}
\item(\textbf{Maximal arcs in Steiner systems with $\alpha=3$}) By Skolem's construction \cite{skolem1958} we have $S(2,3,\theta)$ for all $\theta\geq7$ and $\theta \equiv 1,3 \mod{6}$. Moreover, for all $\theta\geq7$ and $\theta \equiv 3,7 \mod{12}$ there exists a Steiner system $S(2,3,\theta)$ with at least one maximal arc \cite{quattrocchi1993}.
\item(\textbf{Maximal arcs in Steiner systems with $\alpha=4$})  It is known \cite{stinson2004} that Steiner systems with $\alpha = 4$ exist if and only if $\theta \geq 13$ and $$\theta \equiv 1,4\mod{12}.$$
Furthermore, if $\displaystyle \rho=\frac{\theta-1}{3}$ is a prime power, then there exists an Steiner system $S(2, 4, \theta)$ with a maximal arc of size $\rho+1$ \cite{greig2003}.
\end{itemize}
To our best knowledge, there are no other general results about the existence of maximal arcs in Steiner systems with higher values of $\alpha$.
\end{remark}

%

\subsection{FR codes from resolvable designs}
A major drawback of FR codes obtained from Steiner systems is that the repetition degree of the symbols is quite inflexible. In particular, it is not possible to vary the repetition degree and hence the failure resilience of the DSS in an easy way. To address this issue, we now introduce FR codes that are derived from resolvable designs.

A design $(\Omega, V)$ is said to be resolvable if we can divide the blocks in $V$ into equal-sized partitions such that (a) each partition contains all the symbols in $\Omega$, and (b) the blocks in a given partition have no symbols in common. Under certain conditions, these designs also allow for $\beta$-recoverability. A FR code obtained from such a design is called a resolvable FR code and is naturally resilient to any failure pattern that ensures that at least one partition is left intact. In the discussion below, we introduce the notion of a net FR code (a subclass of resolvable FR codes) that ensures $\beta$-recoverability.

Under this overall framework, we construct several families of net FR codes that allow us to vary the repetition degree in an easy manner. We demonstrate that there exist net FR codes with $\beta > 1$ that cannot be derived by trivial $\beta$-expansion. Furthermore, we answer an open question of \cite{el2010} by demonstrating a FR code that cannot be constructed from Steiner systems. We also provide explicit calculations of the file size for certain ranges of $k$. The overall structure of this subsection is as follows. We first introduce our construction, show that it results in a net FR code and then calculate its file size.

\begin{definition}[\textbf{Resolvable FR Code}] \label{resolvable fractional repetition code}  Let $\calC = (\Omega, V)$ where $V = \{V_1, \dots, V_n\}$ be a FR code. A subset $P \subset V$ is said to be a parallel class if for $V_i \in P$ and $V_j \in P$ with $i \neq j$ we have $\displaystyle V_i \cap V_j = \emptyset$ and $\cup_{\{j : V_j \in P\}} V_j = \Omega$. A partition of  $V$ into $r$ parallel classes is called a resolution. If there exists at least one resolution then the code is called a \textit{resolvable FR code}.
\end{definition}
For a resolvable FR code, we call two storage nodes parallel if they belong to the same parallel class and non-parallel otherwise.
The properties of a resolvable FR code are best illustrated by means of the following example.
\begin{example}
Consider a DSS with parameters $\alpha = 3, \theta = \alpha^2 = 9, \rho = 2$ and $\beta = 1$. Suppose that we arrange the symbols in $\Omega = \{1, \dots, 9\}$ in a $\alpha \times\alpha$ array $A$ shown below.
$$A=\begin{array}{ccc}
1&2&3\\
4&5&6\\
7&8&9
\end{array}.$$
Let the rows and the columns of $A$ form the nodes in the FR code $\calC$ (see Fig. \ref{DSS-(6,3,3)}), thus $n=6$. It is evident that there are two parallel classes in $\calC$, $P^r = \{V_1, V_2, V_3\}$ (corresponding to rows) and $P^c = \{V_4, V_5, V_6\}$ (corresponding to columns). As $\rho = 2$, this code can tolerate one failure.

By our construction it is evident that for $V_i \in P^r$ and $V_j \in P^c$, we have $|V_i \cap V_j| = 1$. Using this we can compute the file size $\calM$ when $k=3$, as follows. Let $a + b=3$ with $a \geq b$. Then, the number of distinct symbols in a set of $3$ nodes from $\calC$ is
\begin{equation*}
3a+(3-a)(3-a)= a^2 + 9 - 3a,
\end{equation*}
where $a$ nodes are from $P^r$ and $(3-a)$ nodes are from $P^c$. This is minimized when $a = 2$. Thus, $\calM = 7$ and $\mathcal{R}_{\calC}=\frac{7}{18}$. Note also that the code is optimal with respect to the Singleton bound since $k = \lceil \frac{\calM}{\alpha} \rceil$.\end{example}

\begin{figure*} [t]
\centering
\includegraphics[scale=0.65]{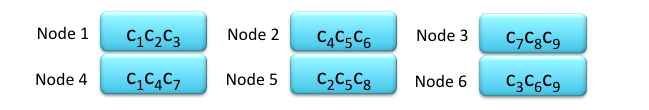}
\caption{A DSS specified with $(n=6,k=3,d=3, \alpha=3)$. Note that the nodes numbered 1,2,3 and 4,5,6 form parallel classes.}
\label{DSS-(6,3,3)}
\end{figure*}
If one starts with a resolvable design with many parallel classes, the repetition degree $\rho$ can be varied easily by adding and/or removing parallel classes if needed. We emphasize that the constructions of \cite{el2010} that are based on Steiner systems largely lack this flexibility as many of them are not resolvable. 


In our proposed systems, we require recovery from a node failure by downloading exactly $\beta$ symbols each from a specified set of $d$ surviving nodes. 
To address this issue, we consider a subclass of resolvable FR codes called net FR codes where the intersection size of any two nodes from distinct parallel classes is exactly $\beta$.

\begin{definition}[\textbf{Net FR Code}] Let $\mathcal{C}=(\Omega,V)$ be a resolvable FR code with parameters $(n=ar, \theta= a^2b, \alpha=ab, \rho=r)$ such that any two non-parallel nodes intersect in exactly $b$ symbols. The design determined by $\mathcal{C}$ is called a net \cite{assmus1992} and we call $\mathcal{C}$ a \textit{net FR code}.
\end{definition}
Examples of net FR codes can be obtained from several combinatorial structures, e.g., grids, affine resolvable designs, Hadamard designs and mutually orthogonal Latin squares (MOLS). We elaborate on these constructions in the subsequent discussion.


Suppose that a net FR code with parameters $(n=ar, \theta= a^2b, \alpha=ab, \rho=r)$ exists. Note that the number of nodes in a parallel class equals $\frac{\theta}{\alpha} = a$. Furthermore, if a given node $V_1 \in V$ fails, this node can be reconstructed by contacting all the nodes in any other intact parallel class and downloading $b$ symbols from each of them. This implies that the code has $d = a, \beta = b$. Next, the code has $\rho = r$ parallel classes and any node can be reconstructed as long as there exists at least one parallel class. Thus, the code is resilient to at least $r-1$ failures, i.e. $\rho_{res} = r-1$.

Note that the parameter $k$ can be chosen such that $1 \leq k \leq d$. The code rate $R_{\mathcal{C}}$ depends on $k$ and needs to be determined. As we shall see determining $R_{\mathcal{C}}$ can be nontrivial in many cases. Specifically, much of the literature in the area of combinatorial designs focuses on pairwise intersections between the storage nodes, whereas the code rate depends on the minimum size of the intersection of any $k$ storage nodes.
Some general results about the code rate of net FR codes can be obtained as discussed in the lemma below. However, a more careful analysis of the algebraic structure of a given construction can allow us to arrive at stronger results.

\begin{lemma}[\textbf{An algorithmic approach for determining the file size of net FR Codes}] 
\label{lemma:net_FR_code}
Let $\mathcal{C}$ be a net FR code with parameters  $(n=ar, \theta= a^2, \alpha=a, \rho=r)$, so that $\beta = 1$. Let $k$ be an integer that satisfies $k \leq \rho$ and $\binom{k-1}{2} < a$. Then, the code rate of the system is $R_{\calC} = (\alpha k -\binom{k}{2})/n\alpha$. 
\end{lemma}
\begin{proof} See Appendix.
\end{proof}
\begin{table}[t]
\begin{center}
\begin{tabular}{|c| c | c | c |}
  \hline
 $\{1,  2,  3,  4\}$ &  $\{5,  6,  7,  8\}$ & $\{9,  10,  11,  12\}$ &  $\{13, 14,  15,  16\}$\\
  \hline
$ \{1,  5,  9,  13\}$ & $\{2,  6,  10,  14\}$ & $\{3,  7,  11,  15\}$ & $\{4,  8,  12,  16\}$\\
 \hline
$\{1,  6,  11,  16\}$ & $\{2,  5,  12,  15\}$ & $\{3,  8,  9,  14\}$ & $\{4,  7,  10,  13\}$\\
 \hline
$ \{1,  7,  12,  14\}$ & $\{2,  8,  11,  13\}$ & $\{3,  5,  10,  16\}$ & $\{4,  6,  9,  15\}$\\
 \hline
\end{tabular}
\end{center}
\caption{\label{Design5} A net FR code with parameters $(16,16,4,4)$.}
\end{table}

\begin{example}\label{4mols}
Consider the following FR code obtained from a net with parameters $(n,\theta, \alpha, \rho) = (16,16,4,4)$. The code arises from mutually orthogonal Latin squares (see Section \ref{sec:grids_hadamard_mols}). This FR code can be specified the nodes presented in Table \ref{Design5}. Each row of the table represents a parallel class. 

Since any two non-parallel nodes intersect in exactly one point, the code corresponds to a net FR code with $a = 4, b = 1$, and $r=4$. Thus, $d=4$ and $\beta = 1$.
Suppose that $k=4$, so that $\binom{k-1}{2} < a$. Our algorithm ({\it cf.} Appendix) may choose the following nodes for $k=4$.
$$L=\{\{1,  2,  3,  4\},\{1,  5,  9,  13\},\{2,  5,  12,  15\}, \{4,  6,  9,  15\}\}.$$
So the file size is $\calM = 10$. However this code is not optimal with respect to Singleton bound. However, observe that the code formed by deleting a parallel class has parameters $(n=12,\theta = 16,\alpha=4,\rho=3)$. In this code any three nodes cover at least $9$ symbols. Thus setting $k= \lceil \frac{9}{4}\rceil = 3$ ({\it cf.} Observation \ref{obs:meet_bound}) results in a code that meets the Singleton bound.

\end{example}

Note that while Lemma \ref{lemma:net_FR_code} applies to all net FR codes with $\beta=1$, the requirement that the storage capacity $\alpha = a \geq \binom{k-1}{2}$ is quite restrictive. 
For certain net FR codes that have a tractable algebraic and/or geometric characterization we can perform a more careful analysis and we now turn our attention to them. 
Our first example is a net where the file size $\calM$ is strictly larger than $k\alpha - \beta\binom{k}{2}$.
\subsubsection{Affine Resolvable FR code}
\label{sec:affine_resolv_fr_code}
Affine resolvable designs are a class of resolvable designs where the intersection between two nodes in different parallel classes can be computed exactly. These can be derived from affine geometries that can be intuitively understood as follows. The set of points corresponds to all elements of $\mathbb{F}_q^n$, the vector space of dimension $n$ over a finite field of size $q$, $\mathbb{F}_q$. Thus, the number of points is $q^n$. The blocks correspond to the solutions of certain sets of linear equations over the vector space. For the sake of simplicity, let us consider just one equation, e.g., $x_1 = a$ for $a \in \mathbb{F}_q$. For each $a\in \mathbb{F}_q$ the solution set is of size $q^{n-1}$. Each such solution set corresponds to a block in the design.
Furthermore, these solution sets partition $\mathbb{F}_q^n$. In a similar manner, one can consider other sets of linear equations of the form $\sum_{i=1}^n b_i x_i = a$ where $b_i \in \mathbb{F}_q$ whose solution sets also partition $\mathbb{F}_q^n$. Furthermore any two linear independent linear equations will have a solution set of size $q^{n-2}$, i.e., the intersection between two such blocks will be exactly $q^{n-2}$. 

The resultant block design is a resolvable design \cite{stinson2004}.  In the discussion below, we present a formal presentation of this idea. We also analyze the file size of the obtained system under the condition that the equations are chosen in a specific manner and for an appropriate range of $k$.

Let $q$ be a prime power, $m \geq 2$ and $\Omega = \mathbb{F}_q^m$. Let $1 \leq \delta \leq m-1$. We treat $\Omega$ as an $m$-dimensional vector space over $\mathbb{F}_q$. A $\delta$-flat is the solution set to a system of $m- \delta$ independent linear equations that can be homogeneous or non-homogeneous. The set $\Omega$ and the set of all $\delta$-flats of $\Omega$ comprise the $m$-dimensional affine geometry over $\mathbb{F}_q$, denoted by $AG_m(q)$. It turns out that one can generate a large class of resolvable designs by considering $AG_m(q)$. Let ${m \brack \delta}_q$ denote the Gaussian coefficient, so that
\begin{align*}
{m \brack \delta}_q = \begin{cases} \frac{(q^m-1) (q^{m-1} - 1) \dots (q^{m-\delta+1} - 1)}{(q^\delta -1)(q^{\delta-1} - 1) \dots (q-1)} & \text{~if~} \delta \neq 0,\\
1 &  \text{~if~} \delta = 0.
\end{cases}
\end{align*}
\begin{theorem}[\textbf{Affine Resolvable Designs}]
\label{thm:resolvable_bibd} \cite{stinson2004} Let $V$ denote the set of all $\delta$-flats in $AG_m(q)$. Then $\Omega = \mathbb{F}_q^m$ and $V$ form a resolvable BIBD with $(\theta = q^m, \rho, \alpha =q^\delta,\lambda)$-BIBD with $n = q^{m-\delta} {m \brack \delta}_q, \rho = {m \brack \delta}_q$ and $\lambda = {m-1 \brack \delta-1}_q$.
\end{theorem}
The case of $m=2, \delta = 1$ corresponds to affine planes. When $\delta = m-1$ we obtain an affine resolvable BIBD with $n = \theta + \rho - 1$. In this case the DSS is specified by the  parameters $\theta=q^m$, $\alpha=q^{m-1}$, $\rho=\frac{q^{m}-1}{q-1}$ and $n=q\rho$.
The design can be obtained by means of the following algorithm.
\begin{itemize}
\item[(i)] Let $\Omega = \{(x_1,x_2, \cdots, x_m): x_i \in \mathbb{F}_q~\mbox{for}~i=1,2,\cdots, m \}$ be the symbol set.
\item[(ii)] Find $\rho$, $(m-1)$-dimensional subspaces of $\bbf_q^m$ such that each of them contains the symbol $(0,0,\cdots, 0) \in \mathbb{F}_q^m$. Note that these subspaces of $\mathbb{F}_q^m$ are the solutions to a single homogeneous linear equation over $\mathbb{F}_q$ in $q$ variables. These $\rho$ subspaces are representatives of the $\rho$ different parallel classes.
\item[(iii)] Construct each parallel class by considering the additive cosets of its representative.
Let $R_1$ be a $(m-1)$-dimensional subspace corresponding to a given homogenous equation. Let $U=\{0, u_1, \dots, u_{q-1}\}$ be the full set of coset representatives of $R_1$. The rest of the blocks can be obtained by the cosets $R^i_1=u_i+R_1$. Note that each of these cosets corresponds to a nonhomogeneous equation.
\end{itemize}
\begin{table}[t]
\begin{center}
\begin{tabular}{|c|}
\hline
$R_1 = \{000, 001, 002, 010, 020, 011, 012, 021, 022\}$\\
 \hline
$ R_2 =\{000, 001, 002, 100, 200, 101, 102, 201, 202\}$ \\
 \hline
$R_3 = \{000, 001, 002, 110, 220, 111, 112, 221, 222\} $ \\
 \hline
$R_4 = \{000, 001, 002, 120, 210, 121, 122, 211, 212\}$ \\
 \hline
$R_5 = \{000, 010, 020, 100, 200, 110, 120, 210, 220\}$ \\
 \hline
 $R_6 = \{000, 010, 020, 101, 202, 111, 121, 212, 222\}$ \\
 \hline
 $R_7 = \{000, 010, 020, 102, 201, 112, 122, 211, 221\}$ \\
 \hline
 $R_8 = \{000, 011, 022, 100, 200, 111, 122, 211, 222\}$ \\
 \hline
 $R_9 = \{000, 011, 022, 101, 202, 112, 120, 210, 221\}$ \\
 \hline
 $R_{10} = \{000, 011, 022, 102, 201, 110, 121, 212, 220\}$ \\
 \hline
 $R_{11} = \{000, 012, 021, 100, 200, 112, 121, 212, 221\} $ \\
 \hline
 $R_{12} = \{000, 012, 021, 101, 202, 110, 122, 211, 220\}$ \\
 \hline
 $R_{13} = \{000, 012, 021, 102, 201, 111, 120, 210, 222\}$\\
 \hline
\end{tabular}
\end{center}
\caption{\label{representatives} Representatives of parallel classes of the FR code with parameters $(n = 39, \theta = 27, \alpha = 9, \beta=3, d=3,  \rho = 13)$.}
\end{table}

\begin{example}[\textbf{An example of an Affine Resolvable Design}] \cite{stinson2004}
Let $q=3$ and $m=3$. The set of symbols is $\Omega = \mathbb{F}_3^3$ and there are 39 blocks which can be partitioned into $13$ parallel classes. The representatives of the $13$ parallel classes are specified in the Table \ref{representatives}, where the vector $[x_1~x_2~x_3]$ is simply written as $x_1 x_2 x_3$.
The other blocks are additive cosets of these 13 representatives. For example, the first parallel class consists of the following blocks.
\begin{align*} B_1 &= \{000, 001, 002, 010, 020, 011, 012, 021, 022\},\\
B_2 &= \{100, 101, 102, 110, 120, 111, 112, 121, 122\}, \text{~and}\\
B_3 &= \{200, 201, 202, 210, 220, 211, 212, 221, 222\}.
\end{align*}
Here the blocks $B_1$, $B_2$ and $B_3$ correspond to equations $x_1=0$, $x_1=1$ and $x_1=2$ respectively.
\end{example}
The overlap between blocks from different parallel classes in the case of affine resolvable designs is known from the following result.
\begin{lemma}
\label{lemma:intersect_affine}
\cite{stinson2004} Any two blocks from different parallel classes of an affine resolvable $(\theta, \rho, \alpha, \lambda)$-BIBD intersect in exactly $\alpha^2/\theta$ symbols.
\end{lemma}
Using the above facts, we can conclude that an affine resolvable BIBD is an instantiation of a net FR code with parameters $(n = q \frac{q^{m}-1}{q-1}, \theta = q^m, \alpha = q^{m-1}, \rho = \frac{q^{m}-1}{q-1})$ and $d = q, \beta = q^{m-2}$. Of course, the repetition degree can be varied by only retaining as many parallel classes as needed.

\begin{remark}[\textbf{Affine Resolvable FR Codes cannot be obtained by trivial $\beta$-expansion}]
It is important to note that the affine resolvable FR codes are an example of a FR code family with $\beta > 1$ that cannot be obtained by replicating the symbols of a smaller code. To show this we will simply use Observation \ref{non_trivial_FR}. Specifically, consider $m\geq 2q+1$ and $q \geq 3$. In this case the affine resolvable FR code will have parameters $\theta = q^m, \alpha = q^{m-1}, \rho = \frac{q^m - 1}{q-1}, n = q \rho$ and $\beta = q^{m-2}$. If it could be generated from a smaller code simply by replication, this would imply that the smaller code had a storage capacity of $q$ and $q^2$ total symbols. This means it has at most $\binom{\theta/\beta}{\alpha/\beta}=\binom{q^2}{q} \leq \bigg{(}\frac{q^2 e}{q}\bigg{)}^q = (eq)^q \leq q^{2q}$ distinct storage nodes. However, in the affine resolvable FR code we have $n=q\frac{q^m-1}{q-1} \geq q\frac{q^{2q+1}-1}{q-1}$ which can be verified to be strictly larger than $q^{2q}$.
\end{remark}

We can determine the file size of a code $\calC$ obtained from some specific affine resolvable designs, for certain ranges of $k$.  We consider two scenarios depending on the relationship between $q$ and $m$.

\begin{itemize}
\item (\textbf{Case 1: $q > m$})\\
We choose the code $\calC$ such that it has $r \geq m$ parallel classes such that the $i$-th parallel class of $\calC$ corresponds to the homogeneous equation $ x_1 + \alpha_i x_2 + \alpha_i^2 x_3 + \dots + \alpha_i^{m-1} x_m = 0$, where $\alpha_i, i = 1, \dots r$ are all non-zero and distinct. Note that the distinctness requirement also enforces that $q > r$. 
The equations obtained in this manner are such that {\it any} $m$ equations are linearly independent \cite{horn2012}.

For this code we analyze the file size for a fixed $k \leq m$. For a given set  of $k$ blocks, denoted $A_i, i = 1, \dots, k$, it is possible that multiple blocks from the same parallel class are chosen; suppose that these blocks come from $l$ distinct parallel classes, numbered without loss of generality as $1, \dots, l$. Let $z_i$ denote the number of blocks from the $i$-th parallel class, so that
\begin{equation*}
z_1 + z_2 + \dots + z_l = k.
\end{equation*}
If we pick $k_1$ blocks each from a different parallel class, we can immediately conclude that the total number of symbols covered is $q^{m-k_1}$, as the parallel classes correspond to linearly independent equations. Using this fact and the inclusion-exclusion principle, we have
\begin{align*}
&|\cup_{i=1}^k A_i| = \sum_{i_1 = 1}^l z_{i_1} q^{m-1} - \sum_{i_1 < i_2} z_{i_1} z_{i_2} q^{m-2} \\
&+ \sum_{i_1 < i_2 < i_3} z_{i_1} z_{i_2} z_{i_3} q^{m-3} + \dots + (-1)^l z_{1} z_{2} \cdots z_{l} q^{m-l}.
\end{align*}
Upon inspection, 
it is clear that
\begin{align}
\label{eq:file_size_affine_resolv}
&q^m \bigg{(}1 - \Pi_{i=1}^l \bigg{(}1 - \frac{z_i}{q}\bigg{)}\bigg{)} =  \sum_{i_1 = 1}^l z_{i_1} q^{m-1} - \sum_{i_1 < i_2} z_{i_1} z_{i_2} q^{m-2} \nonumber \\
&+ \sum_{i_1 < i_2 < i_3} z_{i_1} z_{i_2} z_{i_3} q^{m-3} + (-1)^l z_{1} z_{2} \cdots z_{l} q^{m-l}.
\end{align}
Thus, we need to analyze the minimum value of the LHS of equation (\ref{eq:file_size_affine_resolv}) (over the possibilities for $z_i, i = 1, \dots, l$) to determine the file size. Using the AM-GM inequality, 
we obtain
\begin{align*}
\frac{1}{l} \sum_{i=1}^l \bigg{(}1 - \frac{z_i}{q}\bigg{)} = 1 - \frac{k}{lq} &\geq \bigg{[} \Pi_{i=1}^l \bigg{(}1 - \frac{z_i}{q}\bigg{)} \bigg{]}^{\frac{1}{l}}\\
\implies \bigg{[} 1 - \frac{k}{lq} \bigg{]}^l &\geq \Pi_{i=1}^l \bigg{(}1 - \frac{z_i}{q}\bigg{)}.
\end{align*}
Equality holds in the above equation when all the $z_i$ terms are equal. In addition, we show below that the function
\begin{align*}
h(l) = \bigg{[} 1 - \frac{k}{lq} \bigg{]}^l
\end{align*}
takes its maximum value over the set $l = 1, \dots, k$ when $l = k$. To see this, let $0 < \chi = \frac{k}{q} < 1$, and consider $\log h(l) = l \log (1 - \frac{\chi}{l})$. Now,
\begin{align*}
 \frac{d}{dl} \log h(l) = \log (1 - \frac{\chi}{l}) + \frac{\frac{\chi}{l}}{ 1 - \frac{\chi}{l}}.
\end{align*}
Let $\chi_1 = \frac{\chi}{l}$ and let us study the function $h_1(\chi_1) = \log (1 - \chi_1) + \frac{\chi_1}{1 - \chi_1}$. Clearly $h_1(0) = 0$. The derivative of $h_1(\chi_1)$ is non-negative for $0 < \chi_1 < 1$, since it equals $\frac{\chi_1}{(1- \chi_1)^2}$. This implies that $h_1(\chi_1) \geq 0$ for $0 < \chi_1 < 1$ and therefore $h'(l) \geq 0$ in the range $l = 1, \dots, k$, i.e., it is an increasing function in this range. This implies that the maximum value of $h(l)$ in the range $l = 1, \dots, k$ is obtained when $l = k$ and $z_i = 1$ for all $i$.

We conclude that the minimum value of the LHS of equation (\ref{eq:file_size_affine_resolv}) is obtained when $k = l$ and $z_i = 1, i = 1, \dots, k$ and that the file size is $q^m \bigg{(} 1 - \bigg{(} 1 - \frac{1}{q} \bigg{)}^k\bigg{)}$.
\item (\textbf{Case 2: $q \leq m$})\\
In this case we choose the code $\calC$ so that it has $r \leq m$ parallel classes. The chosen parallel classes are such that they belong to linearly independent equations. Once again, we can analyze the file size when $k \leq m$. Suppose that we choose $l$ parallel classes and let $z_i$ denote the number of blocks chosen from the $i$-th parallel class. Note that in this case $l \geq \lceil \frac{k}{q} \rceil$ and $z_i \leq q$ for all $i = 1, \dots, l$. Proceeding as in Case 1, we can argue that the function
\begin{align*}
h(l) = \bigg{[} 1 - \frac{k}{lq} \bigg{]}^l
\end{align*}
attains its maximum when $l = k$ and $z_i = 1$ for all $i = 1, \dots, k$. Thus, in this case as well the maximum file size is given by $q^m \bigg{(} 1 - \bigg{(} 1 - \frac{1}{q} \bigg{)}^k\bigg{)}$.
\end{itemize}

\subsubsection{Resolvable FR codes from Grids, Hadamard designs and MOLS}
\label{sec:grids_hadamard_mols}


Note that affine resolvable codes have $\beta$ which is a prime power. We now construct families of net FR codes where $\beta = 1$. Overall, the idea here is to relate the existence of these codes to combinatorial structures such as grids (two-dimensional arrays), Hadamard designs and mutually orthogonal Latin squares. While these combinatorial structures have been studied in their own right, their usage in constructing FR codes is new. In particular, our construction from MOLS demonstrates an instance of a FR code that cannot be derived from Steiner systems (answering an open question in \cite{el2010}).

An $a \times a$ \textit{grid} is a FR code that is obtained as follows.
\begin{itemize}
\item Let $\Omega = \{0, \dots, a^2 -1\}$. Create an 2D-array $A$ whose $(i, j)-th$ entry is $a \times i + j$, where $0 \leq i, j \leq a-1$.
\item Each column and each row of $A$ determines a storage node.
\end{itemize}
It is clear that the FR code so obtained is resolvable. Specifically, the set of columns and the set of rows form a resolution. The parameters are $(n=2a,\theta=a^2,\alpha=a,\rho=2)$. Note that $\beta = 1$ as any row and any column intersect in exactly one symbol. Thus, the code so obtained is also a net FR code. 

\begin{lemma}[\textbf{File size of grid FR Codes}] \label{gridsize} Let $\mathcal{C}$ be a net FR code obtained from an $a \times a$ grid. If $k$ is even, the file size $\mathcal{M}$ of $\calC$ is $ka - k^2/4$ and if $k$ is odd, it is $ka - (k^2-1)/4$

\end{lemma}
\begin{proof}  Assume that we choose $s$ nodes from the parallel class corresponding to the rows and $t$ nodes from the parallel class corresponding to the columns such that $s+t=k$. Note that $k \leq d = a$. It is evident that any three nodes have an empty intersection. Thus, applying the inclusion-exclusion principle, we conclude that any $k$ nodes cover exactly $\alpha k-st$ symbols. Next, note that $\alpha k - st = ak - ks + s^2 = (s - k/2)^2 + ka - k^2/4$ which takes the minimum value $ka - k^2/4 + \min((k/2 - \lceil k/2\rceil)^2, (k/2 - \lfloor k/2 \rfloor)^2)$, i.e., it equals $ka - k^2/4$ when $k$ is even and $ka - (k^2-1)/4$ when $k$ is odd.
\end{proof}

The following corollary can be obtained by examining conditions under which $k = \lceil \frac{\mathcal{M}}{\alpha}\rceil$.
\begin{corollary}
\begin{itemize}
\item Let $k=2u$ and $u^2<a$. Then the FR code obtained from $a \times a$ grid is optimal with respect to the Singleton bound.
\item Let $k=2u+1$ and $u(u+1)<a$. Then the FR code obtained from $a \times a$ grid is optimal with respect to the Singleton bound.
\end{itemize}
\end{corollary}

A second construction of affine resolvable designs can be obtained from Hadamard matrices or equivalently difference sets as discussed below. Consider an algebraic group $G$ of order $\theta$ and $D \subseteq G$ such that $|D| = \alpha$, with the property that every nonidentity element of $G$ can be expressed as a difference $d_1-d_2$ of elements of $D$ in exactly $\lambda$ ways. We refer to $D$ as a $(\theta,\alpha,\lambda)$-difference set. 

\begin{lemma} [\textbf{Quadratic Residue Difference Set}] \cite{stinson2004}
Let $q=4a-1\geq 7$ be an odd prime power and $G = \mathbb{F}_q$. Let $D=\{z^2: z \in \mathbb{F}_q,~z\neq0 \}$ 
be the set of quadratic residues. Then $D$ is a $(4a-1, 2a-1, a-1)$-difference set in $(\mathbb{F}_q,+)$, where $+$ denotes the additive operation over $\mathbb{F}_q$.\\
\end{lemma}
For any $g \in G$, we define the \textit{translate} of $D$ by $g+D=\{g+d: d \in D\}$, and define the \textit{development} of $D$ by $\mbox{Dev}(D)=\{g+D: g\in G\}$. If $D$ is a $(\theta,\alpha,\lambda)$-difference set in $G$, then $(G,\mbox{Dev}(D))$ is a $(\theta, \rho,\alpha,\lambda)$-BIBD \cite{stinson2004}.

Let $(\Omega, V )$ be the $(4a-1, 2a-1, 2a-1, a-1)$-BIBD constructed by using a quadratic residue difference set. Let $\infty \notin \Omega$, and define for $V'=\{B\cup \{\infty\}: B \in V\}$. Then it can be shown that $(\Omega \cup \{\infty\}, V' \cup \{\Omega-B:B \in V\})$ is an affine resolvable $(4a, 4a-1, 2a, 2a-1)$-BIBD. Using the equations (\ref{bibd_eq_1}) and (\ref{bibd_eq_2}) this corresponds to a net FR code with parameters $\theta = 4a, \alpha = 2a, \beta = a, d = 2, \rho = 4a-1$ and $n=8a-2$ (see \cite{stinson2004}, Chapter 5). 
\begin{table}[t]
\begin{center}
\begin{tabular}{|c| c |}
  \hline
$ \{\infty, 1,2,4\}$ & $\{0,3,5,6\}$\\
 \hline
$ \{\infty,2,3,5\}$ & $\{1,4,6,0\}$\\
 \hline
$ \{\infty,3,4,6\}$ & $\{2,5,0,1\}$\\
 \hline
$ \{\infty,4,5,0\}$ & $\{3,6,1,2\}$\\
 \hline
$\{\infty,5,6,1\}$ &$\{4,0,2,3\}$\\
 \hline
$ \{\infty,6,0,2\}$& $\{5,1,3,4\}$\\
 \hline
$ \{\infty,0,1,3\}$ & $\{6,2,4,5\}$\\
 \hline
\end{tabular}
\end{center}
\caption{\label{hadamarddesign} Hadamard design obtained from the $(7, 3, 1)$-difference set in $\Omega=\bbf_7$.}
\end{table}

\begin{example}
$D=\{1,2,4\}$ is a $(7, 3, 1)$-difference set in $\Omega=\bbf_7$. We can construct the Fano plane by using the difference set $D$ which is a $(7, 3, 3,1)$-BIBD. By applying the above construction we can construct a FR code with parameters $\theta =8, n = 14,\alpha = 4, \rho = 7$. Corresponding storage nodes are presented in Table \ref{hadamarddesign} where each row of the table represents a parallel class. 
%
%
%
%
%
\end{example}
For this class of codes, $d$ is always 2. However, they offer more flexibility in the choice of $\beta$; unlike affine geometry based codes, we do not require $\beta$ to be a prime power.
\begin{remark}[\textbf{FR Codes derived from Hadamard Designs cannot be obtained by trivial $\beta$-expansion with $\beta=a$}] In addition, they provide another example of a family of FR codes that cannot be obtained by trivial $\beta$-expansion with $\beta=a$. To show this, we use Observation \ref{non_trivial_FR}. Suppose that such a code could be obtained by trivial $\beta$-expansion with $\beta=a$, then the original code would correspond to a FR code with $4$ symbols and storage capacity of $2$. In this case, 
there can be at most $\binom{4}{2} = 6$ nodes. In contrast, the code obtained from the Hadamard design has $8a-2>6$ nodes (as $a \geq 2$). 

Since any two non-parallel nodes share $a$ symbols in common, any $k=2$ nodes cover at least $3a$ symbols where $\alpha=2a$. Moreover, $k=2 = \lceil \frac{3a}{2a}\rceil$. Hence the code is optimal with respect to Singleton bound for $k=2$.
\end{remark}


We now discuss another construction of net FR codes that can be obtained from MOLS.
\begin{definition}[\textbf{Latin Square}] A \textit{Latin square} of order $a$ with entries from a set $\Omega$ with $|\Omega|=a$ is an $a \times a$ array $L$ in which every cell contains an element of $\Omega$ such that every row of $L$ is a permutation of $\Omega$ and every column of $L$ is a permutation of $\Omega$.
\end{definition}


\begin{definition}[\textbf{Orthogonal Latin Squares}] Suppose that $L_1$ and $L_2$ are Latin squares of order $a$ with entries from $\Omega_1$ and $\Omega_2$ respectively (where $|\Omega_1| = |\Omega_2|$). We say that $L_1$ and $L_2$ are orthogonal Latin squares if for every $x \in \Omega_1$ and for every $y \in \Omega_2$ there is a unique cell $(i,j)$ such that $L_1(i,j)=x$ and $L_2(i,j)=y.$
\end{definition}
Equivalently, one can consider the superposition of $L_1$ and $L_2$ in which each cell $(i,j)$ is occupied by the pair $(L_1(i,j), L_2(i,j))$. Then, $L_1$ and $L_2$ are orthogonal if and only if the resultant array has every value in $\Omega_1 \times \Omega_2$. A set of $r$ Latin squares $L_1, \dots, L_r$ of order $a$ are said to be \textit{mutually orthogonal} if $L_i$ and $L_j$ are orthogonal for all $1 \leq i < j \leq r$.

We now demonstrate a procedure of constructing net FR codes from MOLS \cite{yates1936}.
Let $\Omega=\{1, 2, \cdots, a^2\}$, and let $L_1, L_2, \cdots L_{r-2}$ be a set of $r-2$ MOLS of order $a$ ($r-2 \leq a-1$).
\begin{itemize}
\item Arrange the elements of $\Omega$ in a $a \times a$ array $A$. Each row and each column of $A$ corresponds to a storage node (this gives us $2a$ nodes).
\item Note that  $L_i$ takes values in $\{1, \dots, a\}$. Within $L_i$ identify the set of $(i,j)$ pairs where a given value $z \in \{1, \dots, a\}$ appears. Create a storage node by including the entries of $A$ corresponding to the identified $(i,j)$ pairs.
\item Repeat this for each $L_i$ and all $z \in \{1, \dots, a\}$. This creates another $(r-2)a$ storage nodes.
\end{itemize}
Thus, a total of $ra$ storage nodes of size $a$ can be obtained. Of course, one can choose fewer storage nodes if so desired.
\begin{example} \label{latin}
Let $a=4$, and $r = 2$. Then, we have the following construction.
$$A=\begin{array}{cccc}
1&2&3&4\\
5&6&7&8\\
9&10&11&12\\
13&14&15&16
\end{array},$$
$$
L_1=\begin{array}{cccc}
1&2&3&4\\
2&1&4&3\\
3&4&1&2\\
4&3&2&1
\end{array} \mbox{~and}~
 L_2=\begin{array}{cccc}
1&2&3&4\\
3&4&1&2\\
4&3&2&1\\
2&1&4&3
\end{array} .$$
We have the cells $(L_1(i,j), L_2(i,j))$ for $i,j=1,2,3,4$ in a matrix form as follows:
$$\begin{array}{cccc}
(1,1)&(2,2)&(3,3)&(4,4)\\
(2,3)&(1,4)&(4,1)&(3,2)\\
(3,4)&(4,3)&(1,2)&(2,1)\\
(4,2)&(3,1)&(2,4)&(1,3)
\end{array} .$$
As we can see from this matrix, all possible cells are covered by the cells $(L_1(i,j), L_2(i,j))$. Thus $L_1$ and $L_2$ are orthogonal.
We have the parallel classes and corresponding storage nodes illustrated in Example \ref{4mols}.
\end{example}

Note that in describing the above construction we assumed the existence of $r-2$ MOLS. We now discuss the issue of the existence of such structures. If $p$ is a prime number, $m$ is a positive integer, and $N = p^m$ then we can construct $N-1$ mutually orthogonal Latin squares as described below.
\begin{itemize}
\item[(i)] Define $L_a : \bbf_N \times \bbf_N \rightarrow \bbf_N$, by $(r, c) \mapsto ar + c$ (where the addition is over $\bbf_N$)
for all $a\in \bbf_N \setminus \{0\}$. Then, $L_a$ is a Latin square since for a given row $r$ (or column $c$) the column (or row) location of an element $s$ is uniquely specified. 
\item [(ii)] For any $a, b \in \bbf_N \setminus \{0\}$, $L_a$ and $L_b$ are orthogonal since for given ordered pair $(s, t)$ the system $ar+c = s$, $br+ c = t$, determine $r = (a- b)^{-1}(s-t)$ and $c = s-ar$ uniquely.
\end{itemize}
\begin{example}
Let N=3. Then  $\mathbb{F}_3=\{0,1,2\}$, $L_1: x+y$ and $L_2: 2x+y$. The two orthogonal Latin squares of order 3 constructed by the above method are
$$L_1=\begin{array}{ccc}
0&1&2\\
1&2&0\\
2&0&1\\
\end{array},~~\mbox{and}~~
~L_2=\begin{array}{ccc}
0&1&2\\
2&0&1\\
1&2&0\\
\end{array}$$
\end{example}
It turns out that in general, the construction described above produces a net FR code. The parameters are discussed in the following discussion.
\begin{lemma}\label{MOLScons}
The construction procedure described above produces 
a net FR code with $\theta = a^2, n=ra, d=\alpha = a, \rho = r$  where non-parallel nodes intersect in exactly one point.
\end{lemma}
\begin{proof}
It is clear from the construction that $\theta = a^2$ and $n = ra$. Each storage node has $a$ symbols so that $\alpha = a$. We need to show that the code is resolvable. Towards this end, note that it is evident that we obtain a parallel class by considering the nodes corresponding to the rows of $A$ (a similar argument holds for the columns of $A$). Next, the nodes obtained by considering Latin square $L_i$ also form a parallel class, since the set of elements obtained by considering the $(i,j)$ pairs corresponding to $z_1 \in \{1, \dots, a\}$ are distinct from those corresponding to  $z_2 \in \{1, \dots, a\}$, if $z_1 \neq z_2$. As we have $r$ parallel classes, we obtain $\rho = r$. Next, consider the overlap between any two storage nodes belonging to different parallel classes. As $L_i$ and $L_j$ are orthogonal, any entry $(k,l) \in [a]\times[a]$ appears exactly once in the superposition of $L_i$ and $L_j$, which implies that the overlap between storage nodes from different parallel classes corresponding to the $L_i$'s is exactly one element. Similarly, a block from a parallel class corresponding to $L_i$ has exactly one overlap with the blocks corresponding to the rows and columns of $A$.
\end{proof}

\begin{remark}[\textbf{There are FR Codes which can be obtained from MOLS but not from Steiner Systems}] \label{MOLS1}
In general, the construction of orthogonal Latin squares is somewhat involved. However, the celebrated results of \cite{bose1960}, demonstrate the construction of two orthogonal Latin squares for all orders $N \neq 2, 6$.
This immediately allows us to construct net FR codes with the following parameters $n=4a, \theta=a^2,$ $d=\alpha=a,$ $\beta=1,$ and $\rho=4$ for any $a\neq 2, 6$. By applying Lemma \ref{lemma:net_FR_code} we can get the file size $\mathcal{M}=4a-6$ for $k=4$ for $a>6$ and it is optimal with respect to Singleton bound ({\it cf.} Observation 1).

This construction allows us to design some FR codes whose parameters cannot be obtained from Steiner systems. For instance, Let $\alpha=10$ and $\theta=100$. Then to construct a FR code we need use the Steiner system $S(2,10,100)$ which does not exist \cite{lam1989}. However the above construction with two orthogonal Latin squares of order 10 provides us a net FR code with $\alpha=10$ and $\theta=100$.\\
\end{remark}
\begin{lemma}[\textbf{File size of FR Codes obtained from MOLS}]
\label{MOLS2}
Let $p$ be a prime and $m$ be a positive integer, so that there exist $p^m -1$ MOLS of order $p^m$. Consider a subset of these $p^m -1$ MOLS of size $r$ and let $\mathcal{C}$ be a net FR code constructed from them. Then for any $k\leq r$, the code rate $R_{\calC} = (k(p^m)-\binom{k}{2})/n p^m$.
\end{lemma}
\begin{proof}
Let $\eta$ be a primitive element of $\mathbb{F}_{p^m}$. From the construction of the $r$ MOLS, we can associate a set of non-zero field elements $\{\eta^{\alpha_1}, \dots, \eta^{\alpha_r}\}$ so that the $i$-th Latin square is generated by the corresponding $\eta^{\alpha_i}$, where $\alpha_i$'s are distinct.
In the discussion below we demonstrate the existence of $r$ storage nodes that cover exactly $rp^m-\binom{r}{2}$ symbols. The argument will also show the required result for any $k<r$. From the inclusion-exclusion principle it is evident that any $r$ nodes cover at least $rp^m-\binom{r}{2}$ symbols. For demonstrating a set of nodes that cover exactly this number we first pick the storage nodes from different parallel classes and demonstrate that the intersection of any three nodes from this set is empty. 

Towards this end in the $i$-th MOLS, consider the storage node determined by the equation $\eta^{\alpha_i} x + y = \eta^{2\alpha_i}$. This specifies the set of nodes that we will be considering. Three nodes intersect in some symbol if the following system of equations has a solution.
\begin{align}
\label{eq:mols_file_size_1} \eta^{\alpha_i} x+y &= \eta^{2\alpha_i}\\
\label{eq:mols_file_size_2} \eta^{\alpha_j} x+y &= \eta^{2\alpha_j} \\
\label{eq:mols_file_size_3}  \eta^{\alpha_k} x+y &= \eta^{2\alpha_k}
\end{align}

Note that any two equations from the set above are linearly independent and have exactly one solution. Thus, if the above system has a solution, then there exist $\mu \neq 0 $ and $\lambda \neq 0$ such that
\begin{align*}
\lambda \eta^{\alpha_i} + \mu \eta^{\alpha_j} &= \eta^{\alpha_k}\\
\lambda + \mu &= 1\\
\lambda \eta^{2\alpha_i} + \mu \eta^{2\alpha_j} &= \eta^{2\alpha_k}
\end{align*}
Next, we note that it cannot be the case that $\eta^{2\alpha_i} = \eta^{2\alpha_j} = \eta^{2\alpha_k}$. 
To see this note that there are no zero divisors in a finite field so $z_1^2=z_2^2$ implies $z_1=z_2$ or $z_1=-z_2$. Thus, we can conclude that
%
$$\lambda=\frac{\eta^{\alpha_k}-\eta^{\alpha_j}}{\eta^{\alpha_i}-\eta^{\alpha_j}}=\frac{\eta^{2\alpha_k}-\eta^{2\alpha_j}}{\eta^{2\alpha_i}-\eta^{2\alpha_j}}.$$

However $$\frac{\eta^{2\alpha_k}-\eta^{2\alpha_j}}{\eta^{2\alpha_i}-\eta^{2\alpha_j}}=\frac{(\eta^{\alpha_k}-\eta^{\alpha_j})(\eta^{\alpha_k}+\eta^{\alpha_j})}{(\eta^{\alpha_i}-\eta^{\alpha_j})(\eta^{\alpha_i}+\eta^{\alpha_j})}$$
and this implies $\eta^{\alpha_i}=\eta^{\alpha_k}$ which is a contradiction. Thus, a solution to the system of equations in (\ref{eq:mols_file_size_1}) - (\ref{eq:mols_file_size_3})  does not exist. The result follows. 
\end{proof}

\begin{remark} The existence of $p^m-1$ MOLS implies the existence of an affine plane of order $p$ \cite{stinson2004}. Thus choosing $r=p^m-1=k$, we can obtain the corresponding file sizes for affine planes. Codes constructed from affine planes were also considered in \cite{el2010} under Steiner systems.
\end{remark}

\begin{example}
A FR code obtained from affine plane of order $3$ is depicted in Fig. \ref{fig:affine3}. This code can be obtained by following the construction outlined above with $p=3$ and $m=1$. It can be observed that this code is optimal with respect to the Singleton bound when $k=2$. ({\it cf.} Observation \ref{obs:meet_bound}). 
\begin{figure*} [t]
\centering
\includegraphics[scale=0.65]{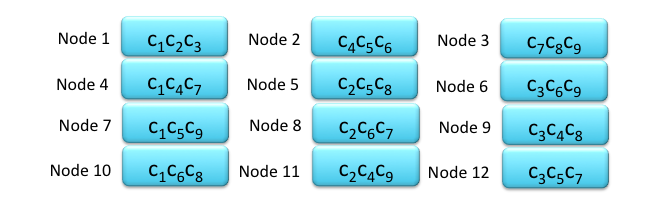}
\caption{\label{fig:affine3} FR code derived from an affine plane of order $3$.}
\end{figure*}
\end{example}

\subsection{Discussion of code parameters achieved by the proposed constructions}
\label{sec:discuss_code_params}
In this subsection, we summarize the range of DSS parameters that our constructions can achieve. Note that there are certain parameter restrictions that any FR code has to satisfy. We list these below. To avoid trivialities, we assume there are no repeated storage nodes in the system.
 $$n\alpha=\theta\rho, \text{~by counting the number of ones in the incidence matrix,}$$
$$n\leq \binom{\theta}{\alpha}, \text{~as the nodes are $\alpha$-sized subsets of the symbols,}$$
$$2 \leq \alpha \leq \theta-1, \text{~as the storage capacity can be at most $\theta -1$,}$$
$$\displaystyle 1 < k \leq d=\frac{\alpha}{\beta}, \text{~as the nodes need to be $\beta$-recoverable,}$$ and
$$\alpha +1 \leq \mathcal{M} \leq k\alpha.$$
If $\beta = 1$, the result of \cite{ernvall12} shows that the conditions are also sufficient for the existence of a FR code; however \cite{ernvall12} does not discuss the file size of such a code. It is evident that specific construction technique imposes additional restrictions. For instance, if the FR code is obtained from a resolvable design, then $\displaystyle \frac{\theta}{\alpha}$ needs to be an integer as it is the number of nodes in a parallel class. In Tables \ref{table:d_greater_k_beta_1} -- \ref{table:local_fr} ({\it cf.} section \ref{sec:back_rel_work}), we summarize the parameters (and the corresponding restrictions that apply) of the different constructions proposed above.  

We emphasize that any FR code is equivalent to a biregular bipartite graph ({\it cf.} Definition \ref{def:bipartite_gr}) and the file size for a given value of $k$ is closely related to the expansion properties of $k$-sized subsets of the storage nodes. It is well recognized that determining the expansion of an arbitrary bipartite graph is a computationally hard problem. In particular, precise numbers are known only for certain families of graphs. High probability results for expansion are known; however, such results are asymptotic in nature and do not provide deterministic constructions. Parameters such the file size can only be found by inspection of the randomly constructed graph. Furthermore, it is not clear whether the $\beta$-recoverability property can be shown for these codes. For these reasons, it is very hard to fully characterize the range of achievable parameters for FR codes (other than the necessary constraints presented above).

Reference \cite{silberstein2014} presents results on the file size of resolvable FR codes that we have considered above. However, we emphasize that Theorem 19 in \cite{silberstein2014} does not apply in our situation. For instance, consider the construction of FR codes from MOLS presented above and Lemma \ref{MOLS2}. Suppose that we choose $r = p^m -1$ and $k = p^m-2$. In this case it can be verified that for large $p$, the result of Theorem 19 in \cite{silberstein2014} does not apply. Furthermore, our affine resolvable design based construction has $\beta > 1$ and the results of \cite{silberstein2014} do not apply here.

On a different note, it can also be argued that one can simply treat the FR codes discussed in this section as local codes, by choosing a value of $k$ that is strictly larger than $d$ (note that $k$ is under our control as a system designer). However, we will now argue that this will result in significantly suboptimal codes with respect to the minimum distance bound in Lemma \ref{local_bound}. 
Suppose for instance that we consider a net FR code with parameters $(n=ar, \theta = a^2b, \alpha = ab, \rho = r)$ with $\beta=b$ and $d = a$. Note that there are $r$ parallel classes in the code. The bound in eq. (\ref{local_bound}), reduces to the Singleton bound as $d \alpha = a^2b = \theta$, so that $\lceil \frac{\mathcal{M}}{d\alpha} \rceil = \lceil \frac{\mathcal{M}}{\theta} \rceil = 1$. Thus, while increasing the value of $k$ above $d$ makes the code local, it will be far from the achieving the local code minimum distance bound in Lemma \ref{local_bound}. As a concrete example, consider a grid code (an instantiation of the net FR code) with $(n = 20,\theta = 100,\alpha = 10,\rho = 2)$. In this case $d=10$ and if $k = 6$, the code is optimal with respect to the Singleton bound as $d_{\min} = 20 - \lceil \frac{51}{10} \rceil + 1 = 15$. However if choose $k = 11$, so that it becomes a local code, the corresponding file size is $\mathcal{M} = 80$, so that the minimum distance bound is $20 - \lceil \frac{80}{10} \rceil + 1 = 13$. However, this code can only recover from at most 9 node failures and not $12$. Thus, such a code is a suboptimal local regenerating code. As all the resolvable codes presented in this section are instances of net FR codes, similar statements apply to all these constructions.


\section{Some characteristics of FR codes obtained from Kronecker Products}
\label{sec:iterative_construction}
The resolvable FR codes derived from affine resolvable designs and Hadamard designs are families of FR codes that have $\beta > 1$ and in many cases cannot be obtained via trivial $\beta$-expansion. In this section, we present the Kronecker product as a technique for obtaining new codes that have $\beta > 1$. 
In essence, we demonstrate the following result. Suppose that we start with a base FR code with storage capacity $\alpha$ where the pairwise intersection between storage nodes is at most one symbol and is such that its file size equals the inclusion-exclusion lower bound in eq. (\ref{eq:inc_enc_lower_bd}). If we consider the Kronecker product of the code with itself, we get a new FR code, where the normalized repair bandwidth equals $\alpha$ and a precise determination of the file size of the new code is possible. FR codes from Steiner systems and their transposes, form a large class of base FR codes that satisfy these requirements. We also demonstrate that the Kronecker product technique yields infinite families of FR codes that cannot be obtained from trivial $\beta$-expansion method. Furthermore, a careful analysis of the construction also allows to conclude that the failure resilience of these codes is as high as possible. We conclude by showing that the property of being resolvable in maintained under taking Kronecker products. 

We begin with a simple example that generates a code that meets the Singleton bound. Let $\theta = 2a+1$ for $a \geq 1$ and the incidence matrices $N_1$ and $N_2$ be equal to $J-I$ where $J$ denotes $\theta \times \theta$ all-ones matrix and $I$ denotes the identity matrix of the appropriate size. Then, the FR code $\mathcal{C}$ obtained from the incidence matrix $\bar{N}=N _1\otimes N_2$  has the following properties:
\begin{itemize}
\item The parameters of the code are $\bar{n}=\bar{\theta}=(2a+1)^2$ and  $\bar{\alpha}=\bar{\rho}=(2a)^2$.
\item A failed node can be recovered by contacting two nodes.
\item Contacting any two nodes recovers at least $2a(2a+1)$ symbols. Thus, when $k=2$, we have that the file size $\mathcal{M} = 2a(2a+1)$, where it can be observed  that $\lceil \frac{\mathcal{M}}{\alpha}\rceil = 2$, so that the code meets the Singleton bound.
\end{itemize}


\begin{example}\label{eg:kron_prod_eg}   Let $\mathcal{C}=(\Omega,V)$ be a FR code with $\Omega=\{1,2,3\}$ and  $V=\{V_1=\{2,3\},V_2=\{1,3\}, V_3=\{1,2\}\}$, so that its incidence matrix $N = \begin{bmatrix} 0 & 1 & 1\\ 1 & 0 & 1\\ 1 & 1 & 0\end{bmatrix}$. The new code is obtained from the incidence matrix of $\bar{N}=N \otimes N$ and the storage nodes are shown in Fig. \ref{fig:kron_local_eg_2}.

Suppose that the outer MDS code has parameters $(9,6)$, so that $\theta = 9, \calM = 6$. In this construction, the file can be recovered by contacting any two nodes, so that $k=2$ and that a failed node can be recovered by contacting two nodes and downloading two packets from each of them. 
\end{example}
 \begin{figure*} [t]
\centering
\includegraphics[scale=0.65]{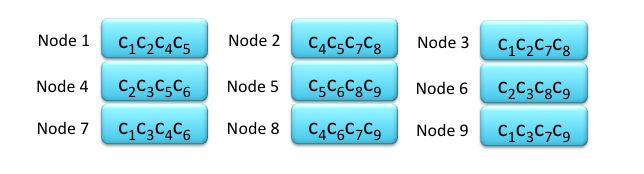}
\caption{\label{fig:kron_local_eg_2} A failed node can be recovered by contacting two nodes and downloading two packets from each of them. The code is resilient up to a total of three failures (corresponding to its minimum distance) and the file size is $6$.}
\end{figure*}

\begin{observation}[\textbf{Non-trivial FR Codes with $\beta>1$ obtained from Kronecker product}] \label{non_trivial_kron} A FR code $\calC$ with parameters $(n, \theta, \alpha, \rho)$, yields a new FR code $\bar{\calC}$ with parameters $(n^2, \theta^2, \alpha^2 \rho^2)$ via Kronecker product  method with itself. If $\alpha$ does not divide $\theta$ then storage nodes of $\bar{\calC}$ cannot be obtained from a trivial $\beta$-expansion with $\beta=\alpha$.
\end{observation}

\begin{example}
Consider the FR code obtained by the Kronecker product of the Fano plane (shown in Fig. \ref{fanoplane1}) with itself. The resultant code will have $49$ symbols with nodes with storage capacity $9$. If this code could be obtained by trivial $\beta$-expansion from a base code with number of symbols $\tilde{\theta}$ and storage capacity $\tilde{\alpha}$, then there has to exist an integer $m$ so that
\begin{align*}
\tilde{\theta} m &= 49, \text{~and}\\
\tilde{\alpha} m & = 9.
\end{align*}
As $9 \nmid  49$, the only feasible solution to the above system of equation is $\tilde{\theta} = 49, \tilde{\alpha} = 9$ and $m = 1$, which corresponds to the Kronecker product code.
\end{example}

In fact, there exists a family of codes whose parameters cannot be obtained via trivial $\beta$-expansion, as discussed in the corollary below.
\begin{corollary} Let $\calC$ be a FR code obtained from a Steiner system $S(2,3,6u+1)$ for some integer $u$. Then the FR code $\bar{\calC}$, which is obtained by the Kronecker product of $\calC$ with itself, cannot be obtained by trivial $\beta$-expansion with $\beta=3$.
\end{corollary}

\begin{lemma}\label{lemma:kron_prod_steiner} Let $\mathcal{C}_1=(\Omega_1,V_1)$ and $\mathcal{C}_2=(\Omega_2,V_2)$ be two FR codes with parameters $(n_1, \theta_1, \alpha, \rho_1)$ and $(n_2, \theta_2, \alpha, \rho_2)$ such that any two storage nodes in $\mathcal{C}_1$ (or $\mathcal{C}_2$) have at most one symbol in common. Let $\mathcal{M}_1$ and $\mathcal{M}_2$ denote the file sizes of  $\mathcal{C}_1$ and $\mathcal{C}_2$ respectively for a given $k_1 \leq \min{ \{n_1,n_2\}}$. Suppose that either $\mathcal{M}_1$ or $\mathcal{M}_2$ is equal to $k_1 \alpha- \binom{k_1}{2}$. Then the FR code $\mathcal{C}$ obtained from Kronecker product of
$\mathcal{C}_1$ and $\mathcal{C}_2$ has parameters $(n = n_1n_2, \theta = \theta_1 \theta_2, \alpha^2, \rho_1\rho_2)$. The file size for $\mathcal{C}$ when $k=k_1$ is given by $k_1 \alpha^2-\alpha \binom{k_1}{2}$. 
\end{lemma}

\begin{proof} Let $N_1$ and $N_2$ denote the incidence matrices of the FR codes $\mathcal{C}_1$ and $\mathcal{C}_2$. Let $c_i$ denote a column in $N_1$ and $d_i$ denote a column in $N_2$. The overlap between any two columns in $N_1 \otimes N_2$ can be expressed as $(c_i \otimes d_j)^t (c_{i'} \otimes d_{j'}) = c_i^t c_{i'} \otimes d_j^t d_{j'} \leq \alpha$. Thus the overlap between any two columns in $N_1 \otimes N_2$ is at most $\alpha$ and therefore the file size of $\mathcal{C}$ is at least $k_1 \alpha^2-\alpha \binom{k_1}{2}$.

We know that any two nodes in $\mathcal{C}_1$ and $\mathcal{C}_2$ have at most one symbol in common. Thus, using a simple inclusion-exclusion principle argument implies that $\mathcal{M}_i \geq k_1 \alpha- \binom{k_1}{2}$ for $i = 1,2$. Furthermore, we are given that one of them meets this lower bound. Without loss of generality we assume that $\mathcal{M}_1 = k_1 \alpha- \binom{k_1}{2}$. This implies that there exists a set of column vectors $\mathcal{I}_1 = \{c_1, \dots, c_{k_1}\}$ in $N_1$ such that they cover $\mathcal{M}_1 = k_1 \alpha- \binom{k_1}{2}$ symbols, i.e., any two columns from $\mathcal{I}_1$ have exactly one symbol in common and any three columns from $\mathcal{I}_1$ have no symbols in common (see Appendix).

Next, we demonstrate a set of columns in $N_1 \otimes N_2$ that meets this lower bound. Let us consider a column in $N_2$, denoted $d_1$ and examine $N_1 \otimes d_1$. Within this set we have a subset of $k_1$ columns denoted $\mathcal{I}_2= \{c_i \otimes d_1$, for $c_i \in \mathcal{I}_1\}$. Now $(c_i \otimes d_1)^t (c_j \otimes d_1) = c_i^t c_j \otimes d_1^T d_1 = \alpha$, whereas any three column vectors from $\mathcal{I}_2$ will have a zero overlap. Thus, the number of symbols covered by this set is exactly $k_1 \alpha^2-\alpha \binom{k_1}{2}$. \end{proof}

This lemma can be used to determine the file size for the Kronecker product of certain Steiner systems.
\begin{lemma}
\label{lemma:max_arc_kron}
Let $\mathcal{C}$ be a FR code obtained from a Steiner system $S(2,\alpha,\theta)$ with $\displaystyle \rho=\frac{\theta-1}{\alpha-1}$ such that it has a maximal arc of size $\rho + 1$. Then the Kronecker product of the transposed code $\mathcal{C}^T$ with itself is such that the file size equals $k\rho^2 - \rho \binom{k}{2}$ for $1 \leq k \leq \rho$.
\end{lemma}
\begin{proof}
The result follows from Lemma \ref{lemma:maximal_arc} and Lemma \ref{lemma:kron_prod_steiner}.
\end{proof}

\begin{remark} By Skolem's construction \cite{skolem1958} we have $S(2,3,\theta)$ for all $\theta\geq7$ and $\theta \equiv 1,3 \mod{6}$. Moreover, for all $\theta\geq7$ and $\theta \equiv 3,7 \mod{12}$ a Steiner system $S(2,3,\theta)$ has at least one maximal arc \cite{quattrocchi1993}. Thus, Lemma \ref{lemma:max_arc_kron} applies. 
\end{remark}

\begin{lemma}\label{Kronec} Let $N_1$ and $N_2$ be incidence matrices of two FR codes such that the size of the pairwise intersection of distinct nodes is at most $1$. Let $(n_1, \theta_1, \alpha, \rho_1)$ and $(n_2, \theta_2, \alpha, \rho_2)$ be parameters of these FR codes respectively. Assume that the FR code obtained from $\bar{N}=N_1 \otimes N_2$ has normalized repair bandwidth $\beta=\alpha$. Then the  FR code $\bar{N}$ is resilient up to $\rho_1\rho_2-1$ failures.
\end{lemma}
\begin{proof} Define $\mathcal{N}(c_i)$ ($\mathcal{N}(d_j)$) to be the set of storage nodes in $N_1$ ($N_2$) that have exactly one symbol in common with $c_i$ ($d_j$). As $N_1$ and $N_2$ are Steiner systems, two nodes have at most one symbol in common. In the discussion below we show that if there are at most $\rho_1 \rho_2 - 1$ failures, we can recover all the nodes. We proceed by contradiction, i.e., assume that there exists a set of failed nodes $F^*$ in $\bar{N}$ with $|F^*| = \rho_1 \rho_2 -1$. Suppose that there is a failed node $c_i \otimes d_j \in F^*$ that cannot be recovered. Note that $\beta = \alpha$. Thus, we need to download $\alpha$ symbols each from the surviving nodes, i.e., we need to consider nodes in $\bar{N}$ that have an overlap of $\alpha$ with $c_i \otimes d_j$.

Our first observation is that only the nodes in $\mathcal{N}(c_i) \otimes d_j$ and $c_i \otimes \mathcal{N}(d_j)$ are useful for recovering $c_i \otimes d_j$. To see this consider a node $c_i' \otimes d_j'$ in $\bar{N}$ such that it does not belong to $\mathcal{N}(c_i) \otimes d_j$ or $c_i \otimes \mathcal{N}(d_j)$. If $c_i' = c_i$, then $d_j' \notin \mathcal{N}(d_j)$, i.e., $(c_i' \otimes d_j')^t (c_i \otimes d_j) = 0$; a similar argument holds when $c_i' \notin \mathcal{N}(c_i), d_j' = d_j$. Otherwise $(c_i' \otimes d_j')^t (c_i \otimes d_j)$ can be at most $1$. Thus, only the nodes in  $\mathcal{N}(c_i) \otimes d_j$ and $c_i \otimes \mathcal{N}(d_j)$ are useful for reconstructing $c_i \otimes d_j$.

Next, note that $c_i$ ($d_j$) can be expressed as the sum of $\alpha$ unit vectors of length $\theta_1$ ($\theta_2$). Let $e_k$ denote the unit vector with a one in the $k$-th location. Thus, $c_i = \sum_{k \in I_1} e_k$, where $I_1 \subset [\theta_1]$ and $d_j = \sum_{l \in I_2} e_l$ where $I_2 \subset [\theta_2]$. Thus, the overlap between $c_i \otimes d_j$ and  $\mathcal{N}(c_i) \otimes d_j$ can be expressed as $e_k \otimes d_j$ for some $k \in I_1$. A similar statement holds for the overlap between $c_i \otimes d_j$ and $c_i \otimes \mathcal{N}(d_j)$. Our next observation is that when we reconstruct $c_i \otimes d_j$, we can either download symbols from $\mathcal{N}(c_i) \otimes d_j$ or from $c_i \otimes \mathcal{N}(d_j)$ but not both. Indeed, for $k \in I_1, l \in I_2$, we have $(c_i \otimes e_l)^t (e_k \otimes d_j) = 1$. Thus, if we download symbols from both $\mathcal{N}(c_i) \otimes d_j$ and from $c_i \otimes \mathcal{N}(d_j)$, then we will need to download strictly more than $\alpha^2$ symbols for reconstructing $c_i \otimes d_j$.

Note that there are $\rho_1$ copies of each $e_k \otimes d_j$, where $k \in I_1$. If there is at least one copy of $e_k \otimes d_j$, for all $k \in I_1$ available in the surviving nodes, then it is clear that $c_i \otimes d_j$ can be recovered by downloading copies of each $e_k \otimes d_j$ from the surviving nodes. Likewise, there are $\rho_2$ copies of each $c_i \otimes e_l$ for $l \in I_2$ and $c_i \otimes d_j$ can be recovered if each of these copies is available in the surviving nodes. In the discussion below we say that $c_i \otimes d_j$ is recoverable if either or both of these situations apply.

Thus, it is clear that if $c_i\otimes d_j$ is not recoverable it has to be the case that all copies of $e_{k^*} \otimes d_j$ for some $k^* \in I_1$ are unavailable. This implies that there exists a set of failed nodes denoted $F_1 \subset \mathcal{N}(c_i) \otimes d_j $ of size at least $\rho_1 - 1$. Arguing in a similar vein, we can consider whether $c_i \otimes d_j$ can be recovered from the nodes in $c_i \otimes \mathcal{N}(d_j)$. 
Based on the discussion above, if $c_i \otimes d_j$ is not recoverable, it has to be the case that there exists a set of failed nodes $F_2 \subset c_i \otimes \mathcal{N}(d_j)$ of size at least $\rho_2 -1$. 
In addition the node sets $\mathcal{N}(c_i) \otimes d_j$ and $c_i \otimes \mathcal{N}(d_j)$ are disjoint, thus $F_1 \cap F_2 = \emptyset$, i.e., it is clear that at least $\rho_1 + \rho_2 - 2$ failures are essential to ensure that $c_i \otimes d_j$ is not recoverable.

Next, we examine whether any of the nodes in $F_1 \cup F_2$ are recoverable. A given node in $F_1$ is of the form $c_{i'} \otimes d_j$ where $c_i^t c_{i'} = 1$ . It is evident that $c_{i'} \otimes d_j$ cannot be recovered from $\mathcal{N}(c_{i'}) \otimes d_j$ as all copies of $e_{k^*} \otimes d_j$ for a specific $k^*$ are unavailable owing to the failure of the nodes in $F_1$. Specifically, note that it rules out the possibility of using the surviving nodes in the set $\mathcal{N}(c_{i}) \otimes d_j$. From the previous observation, it can only be recovered exclusively from the nodes in $c_{i'} \otimes \mathcal{N}(d_j)$.

Thus, there need to be at least $\rho_2 -1$ failures from the node set $c_{i'} \otimes \mathcal{N}(d_j)$ to ensure that $c_{i'} \otimes d_j$ is not recoverable. Furthermore, these failures are distinct from the failures in $F_1 \cup F_2$. Arguing in this way for each node in $F_1$, we conclude that at least $(\rho_1 - 1) (\rho_2 -1)$ failures need to be induced to ensure that none of the nodes in $F_1$ can be recovered.

However, this implies a total of $1 + \rho_1 + \rho_2 - 2 + (\rho_1 - 1)(\rho_2 - 1) = \rho_1 \rho_2 > \rho_1 \rho_2 - 1$ failures. Thus, we conclude that even if an appropriate $F_1 \cup F_2$ can be found for $c_i \otimes d_j$, at least one node in $F_1$ can be recovered. After this recovery, the set $F_1$ cannot exist. This implies that $c_i \otimes d_j$ can be recovered. As the choice of $c_i \otimes d_j$ was arbitrary, we can recover any node when there are at most $\rho_1 \rho_2 -1$ failures.

This bound is tight since each symbol in $\bar{N}$ is repeated $\rho_1\rho_2$ times. Thus, we can easily find a set of $\rho_1 \rho_2$ failures that we cannot recover from.

\end{proof}

\begin{corollary}
Let $N_1$ and $N_2$ be transposes of incidence matrices of two Steiner systems namely $S(2, \alpha_1, \theta_1)$ and $S(2, \alpha_2, \theta_2)$ where the parameters satisfy $\displaystyle \rho=\frac{\theta_1-1}{\alpha_1-1}=\frac{\theta_2-1}{\alpha_2-1}$. Assume the FR code obtained from $\bar{N}$ has normalized repair bandwidth $\beta=\rho$. Then, the FR code is resilient up to $\alpha_1\alpha_2-1$ failures.
\end{corollary}
\begin{proof} Any two nodes meet in exactly one symbol in the FR code obtained by transposes of incidence matrices of a Steiner system. Also note that the main ingredient of the proof Lemma \ref{Kronec} is the property that two nodes meet in at most one symbol in Steiner systems. So the rest follows similarly as in the previous proof.
\end{proof}

We also investigate the properties of FR codes that are generated by taking the Kronecker product of net FR codes with themselves. The Kronecker product does not necessarily produce a new net FR code but it yields a resolvable FR code. For example, in Fig. \ref{fig:complete3} a resolvable FR code is obtained from the Kronecker product of a net FR code with itself. However, the obtained code is not a net FR code. To see this, we note that that node sets $\{1, 5, 9, 13\}$ and $\{2, 6, 10, 14\}$ form parallel classes, but the intersection sizes of node 1 with the nodes in the set $\{2, 6, 10, 14\}$ are either two or zero, which implies that the obtained code is not a net FR code.
\begin{lemma} Let $N$ be the incidence matrix of a net FR code with parameters $(n, \theta, \alpha, \rho)$. Then, the FR code obtained from $\bar{N}=N \otimes N$ is a resolvable FR code.\end{lemma}
\begin{proof} We can order the columns of $N$ with respect to the $\rho$ parallel classes. Assume that the $j$-th block in $i$-th parallel class is represented by the column $c_{i,j}$. We will show for fixed $i$ and $s$, $c_{i,j}\otimes c_{s,r}$ with $1\leq j\leq \frac{\theta}{\alpha}$ and $1\leq r\leq  \frac{\theta}{\alpha}$ forms a set of blocks which is a parallel class. There will be $\frac{\theta^2}{\alpha^2}$ blocks in this set, hence it is enough to show any distinct two blocks does not share any points. Since  $(c_{i,j}^t\otimes c_{s,r}^t)(c_{i,u}\otimes c_{s,v})=(c_{i,j}^tc_{i,u}\otimes c_{s,r}^tc_{s,v})$ equals the zero, the $\frac{\theta^2}{\alpha^2}$ vectors form a parallel class.
\end{proof}

\begin{figure}[t]
\centering
\includegraphics[scale=0.6]{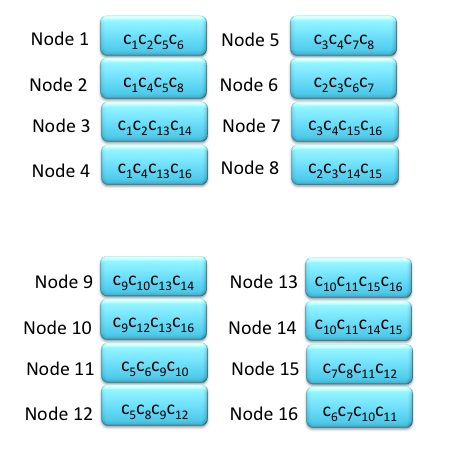}
\caption{\label{fig:complete3} The resultant FR code has $\Omega=\{c_1,c_2,c_3,c_4,c_5,c_6,c_7,c_8,c_9,c_{10},c_{11},c_{12},c_{13},c_{14},c_{15},c_{16}\}$. Each storage node contains $4$ symbols. A failed node can be recovered by contacting two nodes and downloading $2$ packets from each. The code is resilient up to $3$ failures. }
\end{figure}

\begin{example} A simple example can be obtained from $\mathcal{C}=(\Omega,V)$ where $\Omega=\{1,2,3,4\}$ and  $V=\{V_1=\{1,2\},V_2=\{3,4\}, V_3=\{1,3\}, V_4=\{2,4\}\}$. The code obtained from $\bar{N}=N \otimes N$ is illustrated in Fig. \ref{fig:complete3}.
\end{example}



\section{Construction of FR codes when $d < k$}
\label{sec:locally_recoverable_dss}
In the discussion so far, we have considered FR codes where the recovery degree $d \geq k$, i.e., the repair degree ($d$) of the code is at least as high as the number of nodes ($k$) contacted for recovering the file. Of course, the codes operate at the MBR point which implies that they download exactly $\alpha$ symbols for regeneration. However, as discussed in Section \ref{sec:intro}, in many application scenarios it has been recognized that the number of nodes that the new node has to contact is an important metric that needs to be optimized, rather than the repair bandwidth. Note that the definition of a FR code does not rule out codes where $d < k$.

In this section, we discuss constructions of {\it locally recoverable FR codes} that have the property that $d < k$. It turns out that the minimum distance bound for locally recoverable codes that was derived in \cite{gopalan2012, papD12}, needs to be refined for our scenario of exact, uncoded and table-based repair. We derive such a bound and present constructions that meet this bound.

%
%

\begin{definition} [\textbf{Locally recoverable fractional repetition code.}] Let $\mathcal{C} = (\Omega, V)$ be a FR code for a $(n,k,d,\alpha)$-DSS, with repetition degree $\rho$ and normalized repair bandwidth $\beta= \alpha/d$. If the repair degree $d < k$, then the FR code $\calC$ is called a locally recoverable fractional repetition code.
\end{definition}
As before we define $\rho_{res}$ to be the maximum number of node failures such that each failed node can be recovered by contacting $d$ surviving nodes and downloading symbols from them. For a node $V_i \in V$ in $\calC$, let $\calS(V_i) \in V$ denote the set of nodes (with $|\calS(V_i)| < k$) that are contacted if $V_i$ fails. We refer to $\calS(V_i)$ as the local structure associated with $V_i$. Note that it is possible that the set of nodes in $\calS(V_i)$ and the corresponding symbols form a FR code ({\it cf.} Definition \ref{defn:fr_code}); however this is not essential. 



\subsection{Codes for systems with $\rho_{res} = 1$}

Our first construction is a class of codes which is optimal with respect to the bound provided in Lemma \ref{local_bound} and allow local recovery in the presence of a single failure. Our construction leverages the properties of undirected graphs with large girth\footnote{The girth of a graph is the length of its shortest cycle.}. The basic idea is to associate the edges of the undirected graph with the symbols and the vertices with the storage nodes. Each storage node stores its incident symbols. We explain this construction and highlight the intuition behind it by means of the following example.

\begin{example}
The Petersen graph on 10 vertices and 15 edges is a 3-regular graph with girth 5. We label the edges $1, \dots, 10$ and $A, B, \dots, E$ in Fig. \ref{Ptrgraph}. If a given storage node fails, it is evident that it can be regenerated by contacting its corresponding neighbors in the Petersen graph and downloading one symbol each from them. For instance, if node $\{1,A,5\}$ fails, it can download one symbol each from $\{1, B,2\}, \{8,A,9\}$ and $\{4,E,5\}$. Next, note that there is no cycle of length 4, in the Petersen graph. Thus, if we consider any collection of four nodes (as an example), we are guaranteed that the number of edges incident on them is reasonably large. This allows to assert that the file size for such $k=4$ is high. In fact, in the subsequent discussion we show that the file size in this case and for $k=5$ meets the minimum distance bound for locally recoverable codes.
\begin{figure} [t]
\centering
\includegraphics[scale=0.4]{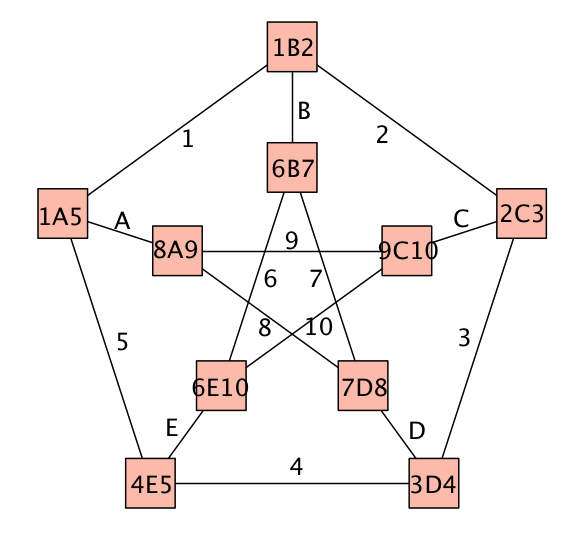}
\caption{{\small The figure shows the Petersen graph with its edges labeled from 1, \dots, 10 and $A, \dots, E$. Each vertex acts as a storage node and stores the symbols incident on it.}}
\label{Ptrgraph}
\end{figure}
\end{example}
We now formalize the basic intuition in the above example, by considering general graphs and precisely calculating the file sizes and minimum distance bounds.

\begin{definition} An undirected graph $\Gamma$ is called an $(s,g)$-graph if each vertex has degree $s$, and the length of the shortest cycle in $\Gamma$ is $g$.
\end{definition}
\begin{construction}
\label{grcons1}
Let $\Gamma=(V',E')$ be a $(s,g)$-graph with $|V'|=n$.
\begin{itemize}
\item[(i)] Arbitrarily index the edges of $\Gamma$ from 1 to $\frac{ns}{2}$.
\item[(ii)] Each vertex of $\Gamma$ corresponds to a storage node and stores the symbols incident on it.
\end{itemize}

\end{construction}
The above procedure yields a FR code $\mathcal{C} = (\Omega, V)$ with $n$ storage nodes, parameters $\theta = \frac{ns}{2}$, $\alpha = s$ and $\rho = 2$. Upon single failure, the failed node can be regenerated by downloading one symbol each from the storage nodes corresponding to the vertices adjacent to it in $\Gamma$ (i.e., $\beta = 1$); thus, the repair degree $d = s$. Note that for this construction, the local structures are typically not FR codes. Suppose that the storage node corresponding to vertex $v_1 \in \Gamma$ fails, then we contact the storage nodes corresponding to its $(s-1)$ neighbors in $\Gamma$; this is the local structure associated with $v_1$. If the girth $g > 3$, then it is clear that the nodes in the local structure do not have symbols in common, i.e., they do not form a FR code.

We note that the work of \cite{el2010} also used the above construction for MBR codes; however, they did not have the girth restriction on $\Gamma$. As we discuss next, $(s,g)$-graphs allow us to construct locally recoverable codes and provide a better bound on the file size when $k \leq g$. We allow the system parameter $k$ to be greater than $d$, however in the work of \cite{el2010}, they consider only the case $k\leq d$. The work of \cite{silberstein2014} also used high-girth graphs, but their constructions are not in the context of locally recoverable codes. 

\begin{lemma} \label{lemma:grcons1_coverage} Let $\mathcal{C} = (\Omega, V)$ be a FR code constructed by Construction \ref{grcons1}. If $s > 2$, and $k \leq g$, we have $|\cup_{i=1}^k V_i| \geq k(s-1)$ for any $V_i \in V, i = 1, \dots k$.
\end{lemma}
\begin{proof}
Let $V_1, V_2,\cdots,V_{k-1}$ and $V_{k}$ be any $k$ nodes in our DSS, where $k \leq g$. We argue inductively. Note that $|V_1| = s > s-1$. Suppose that $|\cup_{i=1}^j V_i| \geq j(s-1) + \xi$ for $j < k$, where $1 \leq \xi \leq j$ is the number of connected components formed by the nodes $V_1, \dots, V_j$ in $\Gamma$. Now consider $|\cup_{i=1}^{j+1} V_i|$ where $j+1 < k$. Note that since $j+1 < g$ there can be no cycle in $\cup_{i=1}^{j+1} V_i$. Thus, $V_{j+1}$ is connected at most once to each connected component in $\cup_{i=1}^j V_i$. Suppose that $V_{j+1}$ is connected to $\ell$ existing connected components in $\cup_{i=1}^j V_i$, where $0 \leq \ell \leq \min (\xi,s)$. 
Then, the number of connected components in $\cup_{i=1}^{j+1} V_i$ is $\xi - \ell + 1$ and the number of new symbols that it introduces is $s - \ell$. Therefore $|\cup_{i=1}^{j+1} V_i| = j(s-1) + \xi + s - \ell = (j+1)(s-1) + \xi -\ell + 1$. This proves the induction step.

Thus, $|\cup_{i=1}^{k-1} V_i| \geq (k-1)(s-1) + \xi_{k-1}$, where $\xi_{k-1}$ is the number of connected components formed by $V_1, \dots, V_{k-1}$. Now consider $\cup_{i=1}^{k} V_i$. Note that there can be a cycle introduced at this step if $k=g$. Now, if $\xi_{k-1} \geq 2$, it can be seen that $V_k$ can only connect to each of the $\xi_{k-1}$ connected components once, otherwise it would imply the existence of a cycle of length strictly less than $g$ in $\Gamma$. Thus, in this case $|\cup_{i=1}^{k} V_i| \geq k(s-1)$. On the other hand if $\xi_{k-1} = 1$, then $V_k$ can connect at most twice to this connected component. In this case again we can observe that $|\cup_{i=1}^{k} V_i| \geq k(s-1)$.
\end{proof}

\begin{lemma} \label{lemma:high_girth_local_bd} Let $\Gamma=(V,E)$ be a $(s,g)$-graph with $|V|=n$ and $s>2$. If $g\geq k=as+b$ such that $s > b \geq a+1$, then $\mathcal{C}$ obtained  from $\Gamma$ by Construction \ref{grcons1} is optimal with respect to the minimum distance bound in Lemma \ref{local_bound} when the file size $\calM=k(s-1)$.
\end{lemma}
\begin{proof} We have 
$$k(s-1)=(as+b)(s-1)=as^2+(b-a)s-b.$$
Since, $s > b\geq a+1$ the following holds.
$$\displaystyle \left\lceil\frac{k(s-1)}{s}\right\rceil=\left\lceil\frac{as^2+(b-a)s-b}{s}\right\rceil=as+(b-a),$$
and
$$\displaystyle  \left\lceil\frac{k(s-1)}{s^2}\right\rceil = \left\lceil\frac{as^2+(b-a)s-b}{s^2}\right\rceil=\left\lceil a+ \frac{(b-a)s-b}{s^2}\right\rceil = a+1 .$$
From Lemma \ref{lemma:grcons1_coverage}, any $k$ nodes cover at least $k(s-1)$ symbols. Thus, the code is minimum distance optimal since
$$\displaystyle \left \lceil\frac{k(s-1)}{s} \right\rceil+ \left\lceil\frac{k(s-1)}{s^2}\right\rceil  = k+1. \text{~({\it cf.} Observation \ref{obs:meet_bound})}$$

\end{proof}
\begin{corollary}\label{corollary_constr_1}
Let $\Gamma=(V,E)$ be a $(s,g)$-graph with $|V|=n$ and $s>2$. If $g\geq s+2$, then $\mathcal{C}$  obtained  from $\Gamma$ by Construction \ref{grcons1} is optimal with respect to the bound in Lemma \ref{local_bound} for file size $\calM=s^2+s-2$.
\end{corollary}
It can be observed that in the specific case of $s=2$, applying Construction \ref{grcons1} results in a DSS where the union of any $k$ nodes has at least $k+1$ symbols. We now discuss some examples of codes that can be obtained  from our constructions.

Sachs \cite{sachs1963} provided a construction which shows that for all $s,g \geq 3$, there exists a $s$-regular graph of girth $g$. Also, explicit constructions of graphs with arbitrarily large girth are known \cite{lazebnik1995}. Using these we can construct infinite families of optimal locally recoverable codes. 


An $(s,g)$-graph with the fewest possible number of vertices, among all $(s,g)$-graphs is called an $(s,g)$-cage and will result in the maximum code rate for our construction. For instance, the $(3,5)$-cage is the Petersen graph. We note here that bipartite cages of girth 6 were used to construct FR codes in \cite{koo2011} though these were not in the context of locally recoverable codes.
An exhaustive survey of cages can be found in \cite{Exoo08}.

\subsection{Codes for systems with $\rho_{res} > 1$}

Our second class of codes are such that the local structures are also FR codes. The primary motivation for considering this class of codes is that they naturally allow for local recovery in the presence of more than one failure as long as the local FR code has a repetition degree greater than two. Thus, in these codes, each storage node participates in one or more local FR codes that allow local recovery in the presence of failures. We motivate the design of these FR codes by means of the following example.
\begin{example}
An example of such a code is shown in Fig. \ref{Fano}. The main idea is to have four FR codes derived from the Fano plane that are supported on disjoint sets of symbols. We refer to each of these FR codes as local structures. Note that if there are at most two failures, the nodes can be regenerated by simply downloading symbols from the corresponding local structures. Moreover, upon inspection, it is not too hard to see that any set of 15 nodes cover at least 17 symbols. Thus, we obtain an instance of a local FR code with $n=\theta =28,\alpha = 3, \rho = 3$ that has $k=15$ and $d = 3$. As $d < k$, this FR code is local.

\begin{figure*} [t]
\centering
\includegraphics[scale=0.6]{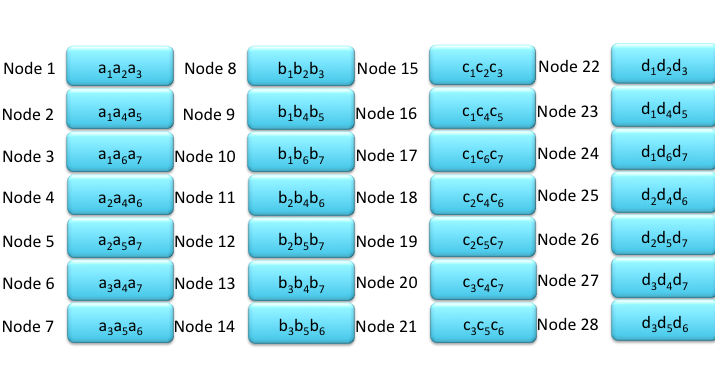}
\caption{The figure shows a DSS where $n=28,k=15,r=3,\theta=28,\alpha =3,\rho =3$ and each local FR code (the columns in the figure) is a projective plane of order $2$ which is also known as a Fano plane. Here, $\rho^{res}=2$. Any set of $15$ nodes cover at least $\mathcal{M}=17$ symbols. Thus, the minimum distance of the code is $14$ when the file size $\mathcal{M}=17$.}
\label{Fano}
\end{figure*}
\end{example}
Note that it is relatively easy to obtain local codes in such a manner, i.e., by considering a collection of FR codes supported on disjoint sets of symbols. However, one really needs to measure them with respect to minimum distance bound for local codes. We did this evaluation for the codes from high girth graphs presented above ({\it cf.} Lemma \ref{lemma:high_girth_local_bd}) and demonstrated that for certain ranges of $k$, the constructed codes were minimum distance optimal. However, we emphasize the minimum distance bound for local codes in Lemma \ref{local_bound} holds for general codes.
In our class of codes, we have the added requirement that each node participates in a local structure that allows it to be recovered by download in case of failure. Accordingly the bound in Lemma \ref{local_bound} is too loose.

For the class of codes that we consider, we derive an upper bound on the minimum distance of such codes when the file size is larger than the number of symbols in one local structure. Following this, we examine (fairly technical) conditions on the local structures that in turn allow for minimum distance optimality of the local FR code. We also demonstrate that several FR codes satisfy these conditions and conclude with some example of minimum distance optimal local FR codes.

\begin{lemma} \label{minimum distance} Let $\mathcal{C}$ be a locally recoverable FR code with parameters $(n,\theta,\alpha,\rho)$ where each node belongs to a local FR code with parameters $(n_{loc},\theta_{loc},\alpha,\rho_{loc})$. Suppose that the file size $\calM > \theta_{loc}$. Then,
\[
\begin{split}
d_{min}& \leq \max \bigg{(}n- \left \lceil \frac{\calM\rho_{loc}}{\alpha} \right \rceil+\rho_{loc},\\
&n+n_{loc}+1- \left \lceil \frac{\calM\rho_{loc}+\theta_{loc}}{\alpha} \right \rceil \bigg{)}.
\end{split}
\]
\end{lemma}

\begin{proof}

%
\begin{figure}[t]
\begin{algorithmic}[1]
\State $S_0=\emptyset$, $i=1$
\While{$H(S_{i-1})<\calM$}
\State For each node $Y_j \in S_{i-1}$, identify a FR code $Pf_j = (\Omega_{Pf_j},V_{Pf_j})$ (if it exists) such that $Y_j \in V_{Pf_j}, V_{Pf_j} \nsubseteq S_{i-1}$. If no such FR code exists, find a FR code that has no intersection with $S_{i-1}$ and set $Pf_1$ equal to it.
\begin{itemize} \item Let $b_j = |\Omega_{Pf_j} \cap H(S_{i-1})|$. Let $j^* = \arg \max_j b_j$. \end{itemize}
\If {$\theta_{loc}-b_{j^*}+H(S_{i-1})<\calM$}
\State Set $S_i = S_{i-1}\cup V_{Pf_{j^*}}$.
\Else \If {there exists $A \subset  V_{Pf_{j^*}}$ such that $|S_{i-1} \cup A| > |S_{i-1}|$ and $H(S_{i-1} \cup A) < \calM$}
        \State Let $V_{Pf^{'}_{j^*}} = \arg \max_{A \subset  V_{Pf_{j^*}}} H(S_{i-1} \cup A) < \calM$. Set $S_i = S_{i-1}\cup V_{Pf^{'}_{j^*}}$.
        \Else
        \State Exit.
        \EndIf
\EndIf
\EndWhile\label{distance while}

\end{algorithmic}
\caption{Algorithm for finding the distance bound}\label{min_dist_algorithm}
\vspace{-0.1in}
\end{figure}

We will apply an algorithmic approach here (inspired by the one used in \cite{gopalan2012}). Namely, we iteratively construct a large enough set $\mathcal{S} \subset V$ so that $|\mathcal{S}| < \calM$. The minimum distance bound is then given by $n - |\mathcal{S}|$. Our algorithm is presented in Fig. \ref{min_dist_algorithm}. Towards this end, let $S_{i}$ and $H(S_i)$ represent the number of nodes and the number of symbols included at the end of the $i$-th iteration. Furthermore, let $s_i=|S_i|-|S_{i-1}|$ and $h_i=|H(S_i)|-|H(S_{i-1})|$, represent the corresponding increments between the $(i-1)$-th and the $i$-th iteration.
We divide the analysis into two cases.
\begin{itemize}
\item {\bf Case 1}: [The algorithm exits without ever entering line 8.] Note that we have $1 \leq s_i \leq n_{loc}$ and $h_i\leq \theta_{loc} -a(n_{loc}-s_i)$ where $a(n_{loc}-s_i)$ is the minimum number of symbols covered by $(n_{loc} - s_i)$ nodes in the local FR code and hence a lower bound on $|\Omega_{Pf_{j^*}} \cap H(S_{i-1})|$. 
    By considering the bipartite graph representing the local FR code ({\it cf.} Definition \ref{def:bipartite_gr}) We see that $\displaystyle a(n_{loc}-s_i) \geq \frac{(n_{loc}-s_i)\alpha}{\rho_{loc}}$ .
    Thus, we have
$$\displaystyle \theta_{loc}-a(n_{loc}-s_i) \leq \theta_{loc}-\frac{n_{loc}\alpha-s_i\alpha}{\rho_{loc}}=\frac{s_i\alpha}{\rho_{loc}}.$$ Suppose that the algorithm runs for $l$ iterations and exits on the $l+1$ iteration. Then
$$\displaystyle \sum_{i=1}^ls_i \geq \frac{\rho_{loc}}{\alpha}\sum_{i=1}^l h_i.$$ Since the algorithm exits without ever entering line $8$, it is unable to accumulate even one additional node. Hence
\begin{align*} \sum_{i=1}^l h_i &\geq \calM-\alpha, \text{~which implies that}\\
\sum_{i=1}^l s_i &\geq \left \lceil \frac{\rho_{loc}}{\alpha}(\calM-\alpha) \right \rceil  \text{~by the integer constraint}.
\end{align*}
Thus, the bound on the minimum distance becomes
$$\displaystyle d_{min}\leq n- \left \lceil\frac{\rho_{loc}M}{\alpha} \right \rceil  +\rho_{loc}.$$
\item {\bf Case 2}: [The algorithm exits after entering line 8.] Note that by assumption, $\calM > \theta_{loc}$. Suppose that the algorithm enters line $5$, $l \geq 1$ times. Now we have $\displaystyle \sum_{i=1}^l h_i \geq \calM-\theta_{loc}$, otherwise we could include another local structure. Hence we need to add nodes so that strictly less than $\displaystyle \calM - \sum_{i=1}^l h_i $ symbols are covered. It can be observed that we can include at least $\displaystyle \left \lceil\frac{\calM-\sum_{i=1}^l h_i}{\alpha} \right \rceil -1$  more nodes. Therefore, the total number of nodes accumulated is
\begin{align*}
&\geq \frac{\rho_{loc}}{\alpha}\sum_{i=1}^l h_i+ \left \lceil\frac{\calM-\sum_{i=1}^l h_i}{\alpha} \right \rceil  - 1\\
&\geq \frac{\rho_{loc}-1}{\alpha}(\calM-\theta_{loc})+\frac{\calM}{\alpha}-1\\
&=\frac{\calM\rho_{loc}+\theta_{loc}}{\alpha}-n_{loc}-1.
\end{align*}
Therefore, we have the following minimum distance bound.
$$\displaystyle d_{min} \leq n+n_{loc}+1- \left \lceil \frac{\calM\rho_{loc}+\theta_{loc}}{\alpha} \right \rceil. $$
The final bound is obtained by taking the maximum of the two bounds obtained above.
\end{itemize}

\end{proof}

The following corollary can be also be established.

\begin{corollary}\label{mincor} Let $\mathcal{C}$ be a locally recoverable FR code with parameters $(n,\theta,\alpha,\rho)$ where each node belongs to a local FR code with parameters $(n_{loc},\theta_{loc},\alpha,\rho_{loc})$. Furthermore, suppose that $\mathcal{C}$ can be partitioned as the union of $\ell$ disjoint local FR codes. If the file size $\calM=t\theta_{loc}+\beta$ for some integer $1 \leq t < \ell$ and $\beta \leq \alpha$, we have $\displaystyle d_{min} \leq n-\left \lceil \frac{\calM\rho_{loc}}{\alpha} \right \rceil+\rho_{loc}.$
\end{corollary}
\begin{proof}
Applying the algorithm in Fig. \ref{min_dist_algorithm} it can be observed that we will never enter line 8, as $\mathcal{C}$ consists of the union of disjoint local FR codes and the file size $\calM = t\theta_{loc}+\beta$. Thus, after accumulating $t$ disjoint local FR codes, the algorithm will exit, yielding the required bound.
\end{proof}
\begin{construction}\label{design}
Let $\mathcal{C} = (\Omega, V)$ be a FR code with parameters $(n,\theta,\alpha,\rho)$ such that any $\Delta$+1 nodes in $V$ cover $\theta$ symbols and for $V_i, V_j \in V$, we have $|V_i \cap V_j| \leq \beta$ when $i \neq j$. We construct a locally recoverable FR code $\bar{\mathcal{C}}$ by considering the disjoint union of $l (>1)$ copies of $\mathcal{C}$. Thus, $\bar{\mathcal{C}}$ has parameters $(ln,l\theta, \alpha, \beta)$. We call $\mathcal{C}$ the local FR code of $\bar{\mathcal{C}}$.
\end{construction}

\begin{lemma} \label{lemma:cond_kron_opt} Let $ \bar{\mathcal{C}}$ be a code constructed by Construction \ref{design} for some $l>1$ such that the parameters of the local FR code satisfy $(\rho-1)\alpha\theta-(\theta+\alpha)(\Delta-1)\beta \geq 0$. Let the file size be $\calM= t\theta+\alpha$ for some $1 \leq t < l$. Then $\bar{\mathcal{C}}$ is optimal with respect to Corollary \ref{mincor}.
\end{lemma}
\begin{proof}
It is evident that $\bar{\mathcal{C}}$ is the disjoint union of $l$ local FR codes. Thus, the minimum distance bound here is $d_{min}\leq ln-\left \lceil \frac{(t\theta+\alpha)\rho}{\alpha}\right \rceil+\rho = (l-t)n.$ The code is optimal when any $tn+1$ nodes in $\bar{\mathcal{C}}$ cover at least  $\calM= t\theta+\alpha$ symbols. We show that this is the case below.

Let $a_i$ be the number of nodes that are chosen from the $i$-th local FR code and $X_i$ be the symbols covered by these $a_i$ nodes. Note that for any $1\leq i\leq l$ if $a_i \geq \Delta+1$, then $X_i = \theta$ (the maximum possible). Suppose there are $0 \leq t_1\leq t$ local FR codes that cover $\theta$ symbols. In this case it suffices to show that  $(t-t_1)n+1$ nodes cover at least  $(t-t_1)\theta+\alpha$ symbols. Here we can omit case of $t=t_1$, since our claim clearly holds in this situation. Suppose that these nodes belong to $s$ local FR codes, where $a_i \leq \Delta, i = 1, \dots, s$. By applying Corradi's lemma \cite{jukna2001} we obtain
{\small
\[
\begin{split}
|X_i|\geq \frac{\alpha^2a_i}{\alpha+(a_i-1)\beta}
&\geq \frac{\alpha^2a_i}{\alpha+(\Delta-1)\beta}.\\
\end{split}
\]
}
This implies that

\[
\begin{split}
&\sum_{i=1}^{s}|X_i| \geq \sum_{i=1}^{s}\frac{\alpha^2a_i}{\alpha+(\Delta-1)\beta}\\
&=\frac{\alpha^2}{\alpha+(\Delta-1)\beta}\sum_{i=1}^{s}a_i\\
&=\frac{\alpha^2}{\alpha+(\Delta-1)\beta}((t-t_1)n+1)\\
&= \frac{(t-t_1)\theta\rho\alpha}{\alpha + (\Delta -1)\beta} + \frac{\alpha^2}{\alpha + (\Delta -1)\beta} \text{{\normalsize~(since $n\alpha = \theta\rho$)}}\\
&= (t-t_1)\theta  + \bigg{(}\frac{\rho\alpha}{\alpha + (\Delta -1)\beta} - 1\bigg{)}(t - t_1)\theta + \frac{\alpha^2}{\alpha + (\Delta -1)\beta}\\
&\geq (t-t_1)\theta + \frac{((\rho-1)\alpha - (\Delta -1)\beta)\theta + \alpha^2}{\alpha + (\Delta -1)\beta}\\
&\geq (t-t_1)\theta + \alpha \text{~{\normalsize (using the assumed conditions).}}
\end{split}
\]

\end{proof}
The above lemma can be used to generate several examples of locally recoverable codes with $\rho_{res} > 1$. We discuss two examples below. 
\begin{example}\label{affine resolvable}
Let $q$ be a prime power. We consider the codes obtained from affine resolvable designs discussed in Section \ref{sec:affine_resolv_fr_code}. These codes have parameters $\theta = q^m, \alpha = q^{m-1}, \rho = \frac{q^m-1}{q-1}$ and $n =q \rho$. These codes are resolvable and hence we can vary the repetition degree by choosing an appropriate number of parallel classes. Note that the number of nodes in a parallel class is $\theta/\alpha = q$.

Suppose we choose the local FR code by including $q^{m-1}$ parallel classes, so that the repetition degree is $q^{m-1}$ and there are $n = q^m$ nodes. Furthermore, since the design is affine resolvable, $\beta = q^{m-2}$. The value of $\Delta$ ({\it cf.} Definition \ref{design}) can be determined as follows. For the local FR code, any subset of at least $q^m - q^{m-1} + 1$ nodes has at least one intact parallel class, which covers all the $\theta = q^m$ symbols. Accordingly, for this code we can conclude that $\Delta = q^m - q^{m-1}$.

Next, we verify the conditions of Lemma \ref{lemma:cond_kron_opt}. For this local FR code, we have that
\begin{align*}
&(\rho-1)\alpha\theta-(\theta+\alpha)(\Delta-1)\beta\\
&= (q^{m-1} - 1)q^{2m-1} - (q^m + q^{m-1})(q^m - q^{m-1} - 1) q^{m-2}\\
&=q^{2m-3}(q^{m+1} - q^2 - (q+1)(q^m - q^{m-1} - 1))\\
&=q^{2m-3}(q^{m+1} - q^2 - (q^m - q^{m-1} - 1) - (q^{m+1} - q^m - q))\\
&=q^{2m-3}(q^{m-1} + q+1 - q^2)\\
&\geq 0, \text{~when $m \geq 3$.}
\end{align*}

Thus, to summarize for the local FR code under consideration, the conditions of Lemma \ref{lemma:cond_kron_opt} apply when $m \geq 3$. Thus, we can construct a FR code by consider the disjoint union of $l$ of these local FR codes using Construction \ref{design}. The code will be optimal with respect to the bound derived in Corollary \ref{mincor} for file sizes of the form $t q^m + q^{m-1}$ for $ 1 \leq t < l$.
\end{example}

\begin{example}\label{projective plane} A projective plane of order $q$ also forms a FR code $\mathcal{C} = (\Omega, V)$, where $\alpha = q+1$ and $\rho = q+1$. Furthermore, $|V_i \cap V_j| = 1$ if $i \neq j$ and each pair of symbols appears in exactly one node; this further implies that $\beta = 1$. A simple counting argument shows that $|\Omega| = \theta = q^2+q+1$ and $n = q^2 + q + 1$. The value of $\Delta$ ({\it cf.} Definition \ref{design}) can be determined in the following manner. Applying Corradi's Lemma, we note that any $q^2 + 1$ nodes cover at least a number of symbols greater than or equal to
\begin{align*}
\frac{(q+1)^2 (q^2 + 1)}{q^2+q+1} &= q^2 + q + \frac{q+1}{q^2+q+1}\\
&> q^2 + q,
\end{align*}
whereby we conclude that $q^2+1$ nodes cover all the $q^2 + q + 1$ symbols. It can also be observed that there is a set of $q^2$ nodes that do not cover all the $q^2+q +1$ symbols as the repetition degree of the symbols is $q+1$. Thus, in this case we can observe that $\Delta = q^2$.


We construct a locally recoverable FR code $\bar{\mathcal{C}}$ by taking $l>1$ copies of the code $\mathcal{C}$. So the code $\bar{\mathcal{C}}$ has parameters $(l(q^2+q+1),l(q^2+q+1),q+1,q+1)$. Let the file size be $\calM= t(q^2+q+1)+q+1$ for some $1\leq t<l$. Then, $\bar{\mathcal{C}}$ is  optimal with respect to Lemma  \ref{minimum distance} and has $\rho_{res} = q$. 
An example is illustrated in Fig. \ref{Fano}.
\end{example}


\section{Conclusions and Future Work}
\label{sec:conclusions_future_work}
In this work we have constructed several classes of fractional repetition codes that can be used in distributed storage systems. These codes allow for a repair process that is exact and uncoded but table-based. Our constructions stem from combinatorial designs such as Steiner systems, affine geometries, Hadamard designs and mutually orthogonal Latin squares. We demonstrate that (i) the repetition degree of the symbols which dictates the failure resilience of the code can be varied in an easy manner, and (ii) construct instances of codes with $\beta > 1$ that cannot be obtained in a trivial manner from codes with $\beta = 1$. In addition, we show that new FR codes can be obtained from taking Kronecker products of existing ones and analyze their properties. For codes with exact, uncoded and local repair property (where $d < k$), we establish an appropriate minimum distance bound and present constructions from high-girth graphs and collections of local FR codes (with specific properties) that meet these bounds. For most of our constructions, we determine the code rate for specific ranges of $k$.

There are several opportunities for future work. It would be interesting to examine applications of designs in other areas of network coding. For instance, \cite{tripathyR15} shows that designs can be used to construct directed acyclic networks that have nontrivial implications for distributed function computation. In principle, several combinatorial designs can be treated as FR codes. However, it would be interesting to examine if there are other families that have desirable properties and lend themselves to an anysis of the system code rate. It is to be noted that the code rate depends on the minimum size of the union of $k$-sized subsets of the storage nodes. It can also be viewed as determining the expansion level of a bipartite graph derived from the incidence matrix of the design. In general, it is somewhat challenging as most results in the literature only discuss pairwise intersections. A related problem would be determine feasible and infeasible parameter ranges for FR codes.

\section{Acknowledgements}
The authors would like to thank the anonymous reviewers whose comments and suggestions significantly improved the quality of the paper.

\section{Appendix}
\noindent {\it Proof of Lemma \ref{lemma:net_FR_code}}.\\
Note that the properties of net FR codes imply that any two storage nodes intersect in either one or zero symbols. Thus,
 $\alpha k-\binom{k}{2}$ is the lower bound on the file size. In the discussion below we demonstrate the existence of $k$ nodes that cover exactly $\alpha k-\binom{k}{2}$ symbols. 
 Let the parallel classes be indexed from $0$ to $\rho - 1$.
\begin{algorithmic}[1]
\State Choose a node $V_0$ from $0$-th parallel class. Initialize $H = \emptyset$, $S = \{V_0\}$ and $i=1$.
\While{$|H| \leq a$ and $|S| < k$}
\State Choose $V_i$ from the $i$-th parallel class such that $V_\ell \cap V_i \notin H$ for all $V_{\ell} \in S$.
\State Set $\displaystyle H = H \bigcup \cup_{V_{\ell} \in S} V_{\ell} \cap V_i$.
\State  Set $S = S \cup V_i$.
\EndWhile
\end{algorithmic}
We need to show that an appropriate $V_i$ can always be chosen in the algorithm and that $|S| = k$ upon exit. To see this note that $H$ tracks the set of pairwise intersections between the nodes at all times. At the beginning of stage $i$, the size of $H$ is at most $\binom{i}{2}$ (by interpreting $\binom{1}{2} = 0$). 
Note that a parallel class has $a$ nodes and that two nodes from the same parallel class do not intersect. Thus, as long as $a > \binom{i}{2}$ we can always find an appropriate $V_i$. By our assumption  $\binom{k-1}{2} < a$. Thus, the algorithm exits with $|S| = k$.
%
%

\begin{lemma}
Consider sets $A_1, \dots, A_k$ such that $|A_i| = \alpha$ and $|A_i \cap A_j| \leq 1$ when $i \neq j$ and $|\cup_{i=1}^k A_i| = k\alpha - \binom{k}{2}$. This implies that $|A_i \cap A_j| = 1$ for $i\neq j$ and $|A_i \cap A_j \cap A_l |=0 $ for all distinct triples $(i,j,l)$ where $i,j,l = 1, \dots, k$.
\end{lemma}
\begin{proof}
By the inclusion-exclusion principle, we have that
\begin{align*}
|\cup_{i=1}^k A_i| \geq \sum_{i} |A_i| - \sum_{i < j} |A_i \cap A_j| \geq k\alpha - \binom{k}{2}.
\end{align*}
However, as $|\cup_{i=1}^k A_i| = k\alpha - \binom{k}{2}$, this implies that $|A_i \cap A_j| = 1$ for all pairs $(i,j)$ such that $i \neq j$.

For a set $I \subseteq [k]$, let $A_I$ denote the set $\cap_{i \in I} A_i$. We note that the given conditions also imply that
\begin{align}\label{eq:inc-ex-residual}
\sum_{\emptyset \neq I \subseteq [k], |I| \geq 3} (-1)^{|I|+1} |A_I| = 0.
\end{align}
We argue that it has to be the case that $|A_I| = 0$ for $\emptyset \neq I \subseteq [k], |I| \geq 3$. Suppose that this is not the case and there are $l$ subsets $I_1, \dots, I_l$ such that $|I_i| \geq 3, i = 1, \dots l$ and $|A_{I_i}| = 1$. For each $I_i$, there has to be a maximal $I^*_i$ such that $I_i \subset I^*_i$. Moreover, it has to hold that $|I^*_i \cap I^*_j| \leq 1$, as otherwise $I^*_i \cup I^*_j$ provides an example of a subset that is larger than both $I^*_i$ and $I^*_j$. This establishes that for each $I_i$, there is a {\it unique} maximal $I^*_i$.

Now, we examine contribution of each of the identified maximal subsets $I^*_i$ to the LHS of eq. (\ref{eq:inc-ex-residual}).
It is evident that $|A_J| = 1$ for all $\emptyset \neq J \subseteq I^*_i$. Let $|I^*_i| = \delta$. This implies that the subset $I^*_i$ induces the following contribution to the LHS of eq. (\ref{eq:inc-ex-residual}): $\sum_{i=3}^\delta (-1)^{i+1}\binom{\delta}{i} = \binom{\delta}{2} - (\delta - 1) > 0$. Thus, the subset $I^*_i$ of maximum cardinality contributes a net positive value to the LHS of eq. (\ref{eq:inc-ex-residual}). Following this we can repeat this argument on the next maximal subset. Note that as the maximal subsets have an intersection of size at most one, each maximal subset contributes the LHS of eq. (\ref{eq:inc-ex-residual}) via distinct terms. Finally, it can be observed that the overall contribution of the maximal subsets accounts for all terms in the LHS of eq. (\ref{eq:inc-ex-residual}). We conclude that if there exist $|A_I| > 0$ for $\emptyset \neq I \subseteq [k], |I| \geq 3$, we have $\sum_{\emptyset \neq I \subseteq [k], |I| \geq 3} (-1)^{|I|+1} |A_I| > 0$, which is a contradiction.

%
%
%

\end{proof}



\begin{thebibliography}{10}
\providecommand{\url}[1]{#1}
\csname url@samestyle\endcsname
\providecommand{\newblock}{\relax}
\providecommand{\bibinfo}[2]{#2}
\providecommand{\BIBentrySTDinterwordspacing}{\spaceskip=0pt\relax}
\providecommand{\BIBentryALTinterwordstretchfactor}{4}
\providecommand{\BIBentryALTinterwordspacing}{\spaceskip=\fontdimen2\font plus
\BIBentryALTinterwordstretchfactor\fontdimen3\font minus
  \fontdimen4\font\relax}
\providecommand{\BIBforeignlanguage}[2]{{%
\expandafter\ifx\csname l@#1\endcsname\relax
\typeout{** WARNING: IEEEtran.bst: No hyphenation pattern has been}%
\typeout{** loaded for the language `#1'. Using the pattern for}%
\typeout{** the default language instead.}%
\else
\language=\csname l@#1\endcsname
\fi
#2}}
\providecommand{\BIBdecl}{\relax}
\BIBdecl

\bibitem{inside_ssd_book}
R.~Micheloni, A.~Marelli, and K.~Eshghi, \emph{Inside Solid State Drives
  (SSDs)}.\hskip 1em plus 0.5em minus 0.4em\relax Springer, 2013.

\bibitem{dimakis2010}
A.~G. Dimakis, P.~B. Godfrey, Y.~Wu, M.~J. Wainwright, and K.~Ramchandran,
  ``Network coding for distributed storage systems,'' \emph{IEEE Trans. on
  Info. Th.}, vol.~56, no.~9, pp. 4539--4551, 2010.

\bibitem{gopalan2012}
P.~Gopalan, C.~Huang, H.~Simitci, and S.~Yekhanin, ``On the locality of
  codeword symbols,'' \emph{IEEE Trans. on Info. Th.}, vol.~58, no.~11, pp.
  6925--6934, 2012.

\bibitem{papD12}
D.~S. Papailiopoulos and A.~G. Dimakis, ``Locally repairable codes,'' in
  \emph{IEEE Intl. Symposium on Info. Th.}, 2012, pp. 2771 --2775.

\bibitem{oggier2011}
F.~Oggier and A.~Datta, ``Self-repairing homomorphic codes for distributed
  storage systems,'' in \emph{Proceedings IEEE INFOCOM}, 2011, pp. 1215--1223.

\bibitem{jiekak2013}
S.~Jiekak, A.-M. Kermarrec, N.~L. Scouarnec, G.~Straub, and A.~V. Kempen,
  ``Regenerating codes: A system perspective,'' \emph{ACM SIGOPS Operating
  Systems Review}, vol.~47, no.~2, pp. 23--32, 2013.

\bibitem{el2010}
S.~E. Rouayheb and K.~Ramchandran, ``Fractional repetition codes for repair in
  distributed storage systems,'' in \emph{48th Annual Allerton Conference on
  Communication, Control, and Computing}, 2010, pp. 1510--1517.

\bibitem{rashmi2009}
K.~V. Rashmi, N.~B. Shah, P.~V. Kumar, and K.~Ramchandran, ``Explicit
  construction of optimal exact regenerating codes for distributed storage,''
  in \emph{47th Annual Allerton Conference on Communication, Control, and
  Computing}, 2009, pp. 1243--1249.

\bibitem{kamath_et_al_14}
G.~Kamath, N.~Prakash, V.~Lalitha, and P.~Kumar, ``Codes with local
  regeneration and erasure correction,'' \emph{IEEE Trans. on Info. Th.},
  vol.~60, no.~8, pp. 4637--4660, 2014.

\bibitem{colbourn2010}
C.~J. Colbourn and J.~H. Dinitz, \emph{Handbook of combinatorial
  designs}.\hskip 1em plus 0.5em minus 0.4em\relax CRC press, 2010.

\bibitem{bose1960}
R.~C. Bose, S.~S. Shrikhande, and E.~T. Parker, ``{Further results on the
  construction of mutually orthogonal Latin squares and the falsity of Euler's
  conjecture},'' \emph{Canad. J. Math}, vol.~12, pp. 189--203, 1960.

\bibitem{stinson2004}
D.~R. Stinson, \emph{Combinatorial designs: construction and analysis}.\hskip
  1em plus 0.5em minus 0.4em\relax Springer, 2004.

\bibitem{Exoo08}
G.~Exoo and R.~Jajcay, ``Dynamic cage survey,'' \emph{The Electronic Journal of
  Combinatorics}, 2008.

\bibitem{olmez2012}
O.~Olmez and A.~Ramamoorthy, ``Repairable replication-based storage systems
  using resolvable designs,'' in \emph{50th Annual Allerton Conference on
  Communication, Control, and Computing (Allerton)}, 2012, pp. 1174--1181.

\bibitem{olmez2013_2}
------, ``Constructions of fractional repetition codes from combinatorial
  designs,'' in \emph{47th Asilomar Conf. on Signals, Systems and Computers},
  2013, pp. 647--651.

\bibitem{olmez2013}
------, ``Replication based storage systems with local repair,'' in
  \emph{International Symposium on Network Coding (NetCod)}, 2013, pp. 1--6.

\bibitem{rashmi2011}
K.~V. Rashmi, N.~B. Shah, and P.~V. Kumar, ``{Optimal exact-regenerating codes
  for distributed storage at the MSR and MBR points via a product-matrix
  construction},'' \emph{IEEE Trans. on Info. Th.}, vol.~57, no.~8, pp.
  5227--5239, 2011.

\bibitem{SuhR11}
C.~Suh and K.~Ramchandran, ``{Exact-Repair MDS Code Construction Using
  Interference Alignment},'' \emph{IEEE Trans. on Info. Th.}, vol.~57, no.~3,
  pp. 1425 --1442, 2011.

\bibitem{tian2013}
C.~Tian, V.~Aggarwal, and V.~A. Vaishampayan, ``Exact-repair regenerating codes
  via layered erasure correction and block designs,'' in \emph{IEEE Intl.
  Symposium on Info. Th.}, 2013, pp. 1431--1435.

\bibitem{papailiopoulos2012}
D.~S. Papailiopoulos, J.~Luo, A.~G. Dimakis, C.~Huang, and J.~Li, ``Simple
  regenerating codes: Network coding for cloud storage,'' in \emph{Proceedings
  IEEE INFOCOM}, 2012, pp. 2801--2805.

\bibitem{shah2012}
N.~B. Shah, K.~V. Rashmi, P.~V. Kumar, and K.~Ramchandran, ``Interference
  alignment in regenerating codes for distributed storage: necessity and code
  constructions,'' \emph{IEEE Trans. on Info. Th.}, vol.~58, no.~4, pp.
  2134--2158, 2012.

\bibitem{tamo2011}
I.~Tamo, Z.~Wang, and J.~Bruck, ``{MDS array codes with optimal rebuilding},''
  in \emph{IEEE Intl. Symposium on Info. Th.}, 2011, pp. 1240--1244.

\bibitem{papailiopoulos2011}
D.~S. Papailiopoulos, A.~G. Dimakis, and V.~R. Cadambe, ``Repair optimal
  erasure codes through hadamard designs,'' in \emph{49th Annual Allerton
  Conference on Communication, Control, and Computing}, 2011, pp. 1382--1389.

\bibitem{shah_et_al12_rep_by_transfer}
N.~B. Shah, K.~V. Rashmi, P.~V. Kumar, and K.~Ramchandran, ``Distributed
  storage codes with repair-by-transfer and nonachievability of interior points
  on the storage-bandwidth tradeoff,'' \emph{IEEE Trans. on Info. Th.},
  vol.~58, no.~3, pp. 1837--1852, March 2012.

\bibitem{rawatKSV14}
A.~S. Rawat, O.~O. Koyluoglu, N.~Silberstein, and S.~Vishwanath, ``Optimal
  locally repairable and secure codes for distributed storage systems,''
  \emph{IEEE Trans. on Info. Th.}, vol.~60, no.~1, pp. 212--236, 2014.

\bibitem{ShumH12}
K.~W. Shum and Y.~Hu, ``Functional-repair-by-transfer regenerating codes,'' in
  \emph{IEEE Intl. Symposium on Info. Th.}, July 2012, pp. 1192--1196.

\bibitem{HuLS13}
Y.~Hu, P.~P.~C. Lee, and K.~W. Shum, ``Analysis and construction of functional
  regenerating codes with uncoded repair for distributed storage systems,'' in
  \emph{Proceedings IEEE INFOCOM}, 2013, pp. 2355--2363.

\bibitem{koo2011}
J.~C. Koo and J.~T. Gill, ``Scalable constructions of fractional repetition
  codes in distributed storage systems,'' in \emph{49th Annual Allerton
  Conference on Communication, Control, and Computing}, 2011, pp. 1366--1373.

\bibitem{ernvall12}
T.~Ernvall, ``The existence of fractional repetition codes,'' 2012, [Online]
  Available: http://http://arxiv.org/abs/1201.3547.

\bibitem{silberstein2014}
N.~Silberstein and T.~Etzion, ``Optimal fractional repetition codes based on
  graphs and designs,'' \emph{IEEE Trans. on Info. Th.}, vol.~61, no.~8, pp.
  4164--4180, 2015.

\bibitem{kamath2013}
G.~M. Kamath, N.~Silberstein, N.~Prakash, A.~S. Rawat, V.~Lalitha, O.~O.
  Koyluoglu, P.~V. Kumar, and S.~Vishwanath, ``{Explicit MBR all-symbol
  locality codes},'' in \emph{IEEE Intl. Symposium on Info. Th.}, 2013, pp.
  504--508.

\bibitem{quattrocchi1993}
G.~Quattrocchi and H.~Zeitler, ``Hyperovals in steiner triple systems,''
  \emph{Journal of Geometry}, vol.~47, no.~1, pp. 125--130, 1993.

\bibitem{greig2003}
M.~Greig and A.~Rosa, ``Maximal arcs in steiner systems $s(2, 4, v)$,''
  \emph{Discrete Mathematics}, vol. 267, no.~1, pp. 143--151, 2003.

\bibitem{assmus1992}
E.~F. Assmus, \emph{Designs and their Codes}.\hskip 1em plus 0.5em minus
  0.4em\relax Cambridge University Press, 1992.

\bibitem{skolem1958}
T.~Skolem, ``{Some Remarks on the Triple Systems of Steiner.}''
  \emph{Mathematica Scandinavica}, vol.~6, pp. 273--280, 1958.

\bibitem{horn2012}
R.~A. Horn and C.~R. Johnson, \emph{Matrix Analysis}.\hskip 1em plus 0.5em
  minus 0.4em\relax Cambridge University Press, 2012.

\bibitem{yates1936}
F.~Yates, ``A new method of arranging variety trials involving a large number
  of varieties,'' \emph{The Journal of Agricultural Science}, vol.~26, no.~03,
  pp. 424--455, 1936.

\bibitem{lam1989}
C.~W. Lam, L.~Thiel, and S.~Swiercz, ``The non-existence of finite projective
  planes of order 10,'' \emph{Canad. J. Math}, vol.~41, no.~6, pp. 1117--1123,
  1989.

\bibitem{sachs1963}
H.~Sachs, ``Regular graphs with given girth and restricted circuits,''
  \emph{Journal of the London Mathematical Society}, vol.~1, no.~1, pp.
  423--429, 1963.

\bibitem{lazebnik1995}
F.~Lazebnik and V.~A. Ustimenko, ``Explicit construction of graphs with an
  arbitrary large girth and of large size,'' \emph{Discrete Applied
  Mathematics}, vol.~60, no.~1, pp. 275--284, 1995.

\bibitem{jukna2001}
S.~Jukna, \emph{Extremal combinatorics}.\hskip 1em plus 0.5em minus 0.4em\relax
  Springer, 2001.

\bibitem{tripathyR15}
A.~S. Tripathy and A.~Ramamoorthy, ``Capacity of sum-networks for different
  message alphabets,'' in \emph{IEEE Intl. Symposium on Info. Th.}, 2015, pp.
  606--610.

\end{thebibliography}
\begin{IEEEbiography}{Oktay Olmez} received his Ph.D. in pure mathematics at the Iowa State University under the supervision of Dr. Sung Song in 2012. He also worked as a postdoctoral fellow in the Department of Mathematics and Department of Electrical and Computer Engineering at the Iowa State University between 2012 and 2013. He is currently an Associate Professor in the Department of Mathematics at Ankara University. His research interest include regenerating codes for distributed storage systems, highly regular graphs arising from finite geometries, highly nonlinear boolean functions and construction of combinatorial block designs via difference sets.
\end{IEEEbiography}
\begin{IEEEbiography} {Aditya Ramamoorthy} (M'05) received the B.Tech. degree in electrical engineering
from the Indian Institute of Technology, Delhi, in 1999, and the M.S.
and Ph.D. degrees from the University of California, Los Angeles (UCLA), in
2002 and 2005, respectively.
He was a systems engineer with Biomorphic VLSI Inc. until 2001. From 2005
to 2006, he was with the Data Storage Signal Processing Group of Marvell
Semiconductor Inc. Since fall 2006, he has been with the Electrical and Computer Engineering Department at Iowa State University,
Ames, IA 50011, USA. His research interests are in the areas of network information theory,
channel coding and signal processing for bioinformatics and nanotechnology.

Dr. Ramamoorthy is the recipient of the 2012 Iowa State University's Early Career Engineering Faculty Research Award, the 2012 NSF CAREER award, and the Harpole-Pentair professorship in 2009 and 2010. He served as an associate editor for the IEEE Transactions on Communications from 2011 -- 2014.
\end{IEEEbiography}

\end{document}